\documentclass{article}

\usepackage{fullpage}

\usepackage{epsf}
\usepackage{fancyheadings}
\usepackage{graphics}
\usepackage{graphicx}
\usepackage{psfrag}

\usepackage{color}

\usepackage{amsthm}
\usepackage{amsfonts}
\usepackage{amsmath}
\usepackage{amssymb}
\usepackage{mathrsfs}

\usepackage{accents}
\usepackage{listings}

\usepackage{caption}
\usepackage{subfigure}
\usepackage[linesnumbered, ruled, vlined]{algorithm2e}

\usepackage[titletoc,toc,title]{appendix}

\theoremstyle{plain}

\newtheorem{theorem}{Theorem}
\newtheorem{remark}{Remark}

\newtheorem{lemma}{Lemma}

\newtheorem{definition}{Definition}
\newtheorem{example}{Example}
\newtheorem{fact}{Fact}

\newlength{\widebarargwidth}
\newlength{\widebarargheight}
\newlength{\widebarargdepth}

\makeatletter
\long\def\@makecaption#1#2{
        \vskip 0.8ex
        \setbox\@tempboxa\hbox{\small {\bf #1:} #2}
        \parindent 1.5em  
        \dimen0=\hsize
        \advance\dimen0 by -3em
        \ifdim \wd\@tempboxa >\dimen0
                \hbox to \hsize{
                        \parindent 0em
                        \hfil 
                        \parbox{\dimen0}{\def\baselinestretch{0.96}\small
                                {\bf #1.} #2
                                } 
                        \hfil}
        \else \hbox to \hsize{\hfil \box\@tempboxa \hfil}
        \fi
        }
\makeatother

\long\def\comment#1{}

\newcommand{\matsnorm}[2]{|\!|\!| #1 | \! | \!|_{{#2}}}
\newcommand{\vecnorm}[2]{\| #1\|_{#2}}

\newcommand{\inprod}[2]{\ensuremath{\langle #1 , \, #2 \rangle}}

\newcommand{\kull}[2]{\ensuremath{D_{\mathsf{KL}}(#1\; \| \; #2)}}

\DeclareMathOperator{\var}{var}

\newcommand{\NORMAL}{\ensuremath{\mathcal{N}}}
\newcommand{\BER}{\ensuremath{\mbox{Ber}}}

\newcommand{\Oh}{\ensuremath{\mathcal{O}}}

\newcommand{\Xspace}{\ensuremath{\mathcal{X}}}

\newcommand{\Kspace}{\ensuremath{\mathcal{K}}}
\newcommand{\Rspace}{\ensuremath{\mathcal{R}}}

\newcommand{\boldb}{\ensuremath{\mathbf{b}}}
\newcommand{\bolda}{\ensuremath{\mathbf{a}}}
\newcommand{\xvec}{\ensuremath{\mathbf{x}}}

\newcommand{\evec}[1]{\ensuremath{\widehat{\mathbf{e}}}}

\newcommand{\EE}{\ensuremath{\mathbb{E}}}

\newcommand{\Xhat}{\ensuremath{\mathbf{\widehat{X}}}}
\newcommand{\Yhat}{\ensuremath{\mathbf{\widehat{Y}}}}
\newcommand{\Phat}{\ensuremath{\widehat{P}}}

\newcommand{\yvec}{\ensuremath{\mathbf{y}}}
\newcommand{\zvec}{\ensuremath{\mathbf{z}}}

\newcommand{\1}{\ensuremath{\mathsf{(i)}}}
\newcommand{\2}{\ensuremath{\mathsf{(ii)}}}

\newcommand{\Ind}{\ensuremath{\mathbb{I}}}
\newcommand{\reals}{\ensuremath{\mathbb{R}}}

\newcommand{\fstar}{\ensuremath{f^*}}
\newcommand{\pstar}{\ensuremath{p^*}}
\newcommand{\xstar}{\ensuremath{\xvec^*}}
\newcommand{\ystar}{\ensuremath{\yvec^*}}
\newcommand{\xtilde}{\ensuremath{\widetilde{\xvec}}}

\newcommand{\kstar}{\ensuremath{j^*}}

\newcommand{\quickfigure}[4]{
  \begin{figure}[htbp]
  \centering
  \includegraphics[width=#4]{#3}
  \caption{#2}
  \label{#1}
  \end{figure}
}

\newcommand{\quadfigureexterior}[3]{
\begin{figure}[htbp]
 #3
 \caption{#2}
 \label{#1}
\end{figure}
}

\newcommand{\multifigureexterior}[3]{\quadfigureexterior{#1}{#2}{#3}}

\newcommand{\subfigl}[4]{\subfigure[#3]{\includegraphics[width=#1]{#4} \label{#2}}}

\usepackage{booktabs} 

\begin{document}
\begin{center}

{\bf{\LARGE{Robust Commitments and Partial Reputation}}}

\vspace*{.2in}

{\large{
\begin{tabular}{cc}
Vidya Muthukumar & Anant Sahai \\
\end{tabular}}}
\vspace*{.2in}

\begin{tabular}{c}
Department of Electrical Engineering and Computer Sciences, UC Berkeley \\
\{vidya.muthukumar,sahai\}@eecs.berkeley.edu
\end{tabular}

\vspace*{.2in}

\today

\end{center}
\vspace*{.2in}

\begin{abstract}
Agents rarely act in isolation -- their behavioral history, in particular, is public to others.
We seek a non-asymptotic understanding of how a leader agent should shape this history to its maximal advantage, knowing that follower agent(s) will be learning and responding to it.
We study Stackelberg leader-follower games with finite observations of the leader commitment, which commonly models security games and network routing in engineering, and persuasion mechanisms in economics.
First, we formally show that when the game is not zero-sum and the vanilla Stackelberg commitment is mixed, it is not robust to observational uncertainty.
We propose observation-robust, polynomial-time-computable commitment constructions for leader strategies that approximate the Stackelberg payoff, and also show that these commitment rules approximate the maximum obtainable payoff (which could in general be greater than the Stackelberg payoff).
\end{abstract}

\section{Introduction}

Consider a selfish, rational agent (designated as the ``leader") who is interacting non-cooperatively with another selfish, rational agent (designated as the ``follower").
If the agents are interacting \textit{simultaneously}, and know the game that they are playing, they will naturally play a Nash equilibrium(a) of the two-player game.
This is the traditionally studied solution concept in game theory.
Now, let's say that the leader has the ability to reveal its strategy in advance -- in the form of a \textit{mixed strategy commitment}, and the follower has the ability to observe this commitment and respond to it.
The optimal strategy for the leader now corresponds to the Stackelberg equilibrium of the two-player leader-follower game.
The Stackelberg solution concept enjoys application in several engineering settings where commitment is the natural framework for the leader, such as security~\cite{paruchuri2008playing}, network routing~\cite{roughgarden2004stackelberg} and law enforcement~\cite{muthukumar2017fundamental}.
The solution concept is interpreted in a broader sense as the ensuing strategy played by a patient leader that wishes to build a \textit{reputation} by playing against an infinite number of myopic followers~\cite{kreps1982reputation,milgrom1982predation,fudenberg1989reputation,fudenberg1992maintaining}.
Crucially, one can show rigorously that the leader will benefit significantly from commitment power, i.e. its ensuing Stackelberg equilibrium payoff is at least as much as the simultaneous equilibrium payoff~\cite{von2010leadership}.
Further, several \textit{mechanism design} problems that involve private information revelation - this includes signalling games~\cite{crawford1982strategic} and persuasion games~\cite{kamenica2011bayesian} -  can be thought of as Stackelberg games, and the optimal mechanism can be interpreted as the Stackelberg commitment\footnote{albeit sometimes with multiple followers.}.

But whichever interpretation one chooses, the Stackelberg solution concept assumes a very idealized setting (even over and above the assumptions of selfishness and infinite rationality) in which the mixed strategy commitment is \textit{exactly revealed} to the follower.
Further, the follower 100\% believes that the leader will actually stick to her commitment.
What happens when these assumptions are relaxed?
What if the leader could only demonstrate her commitment in a finite number of interactions -- how would she modify her commitment to maximize payoff, and how much commitment power would she continue to enjoy?
Is she even incentivized to help the follower estimate her commitment effectively?

What changes in a finite-interaction regime is that the follower only observes a part of the leader behavioral history and needs to \textit{learn about the leader's strategic behavior} -- to the extent that he is able to respond as optimally as possible.
By restricting our attention to finitely repeated play, we arrive at a problem setting that is fairly general: these follower agents in general will not know about the preferences of the leader agent.
When provided with historical context, we assume that they will use it rather than ignore it.
A broad umbrella of problems that has received significant attention in the machine learning literature is learning of strategic behavior from samples of play; for example, learning the agent's utility function through \textit{inverse reinforcement learning}~\cite{ziebart2008maximum}, learning the agent's level of rationality~\cite{waugh2013computational}, and inverse game theory~\cite{kuleshov2015inverse}.
While significant progress has been made in this goal of learning strategic behavior, attention has been restricted to the \textit{passive learning setting} in which the leading agent is unaware of the presence of the learner, or agnostic to the learner's utility.
In many situations, the agent himself will be invested in the outcome of the learning process.
In this paper, we put ourselves in the shoes of an agent who is shaping her historical context and is aware of the learner's presence as well as preferences, and study her choice of optimal strategy\footnote{It is worth mentioning the recent paradigm of \textit{cooperative inverse reinforcement learning}~\cite{hadfield2016cooperative} which studies the problem of agent investment in principal learning where the incentives are not completely aligned, but the setting is cooperative. In contrast, we focus on non-cooperative settings.}.
As we will see, the answer will depend on her utility function itself, as well as what kind of response she is able to elicit from the learner.

\subsection{Related work}

The Stackelberg solution concept is used in the engineering and economics literature to model a number of scenarios.
For one, the \textit{Stackelberg security game} is played between a defender, who places different utility levels on different targets to be protected and accordingly uses her resources to defend some subset of targets; and an attacker, who observes the defender strategy and wishes to attack some of the targets depending on whether he thinks they are left open as well as how much he values those targets.
Stackelberg games can also be modeled with a single leader and multiple followers, such as in computer network routing applications~\cite{roughgarden2004stackelberg}.
Many mechanism design problems involve computing an optimal mechanism to commit to, or an optimal way of \textit{revealing private information} - this includes auctions and, more recently, Bayesian persuasion mechanisms~\cite{kamenica2011bayesian}.

Economists have established an important link between the Stackelberg solution concept and the asymptotic limit of \textit{reputation building}.
Reputation effects were first observed in the \textit{chain-store paradox}, a firm-competition game where an \textit{incumbent} firm would often deviate from Nash equilibrium behavior and play its aggressive Stackelberg (pure, in this case) strategy~\cite{selten1978chain}. 
Theoretical justification was provided for this particular application~\cite{kreps1982reputation,milgrom1982predation} by modeling a $(N+1)$-player interaction between a leader and multiple followers, and studying the Nash equilibrium of an ensuing game as $N \to \infty$.
It was shown that the leader would play its \textit{pure Stackelberg strategy}\footnote{This is clearly more restrictive than the mixed strategy Stackelberg solution concept, and not necessarily advantageous over Nash, but it turns out to be so in the firm competition case.} in the Bayes-Nash equilibrium of this game endowed with a common prior on the leader's payoff structure\footnote{A more nuanced model considered a Bayesian prior over a leader being constrained to play its pure Stackelberg strategy as opposed to unconstrained play.}.
This model was generalized to such leader-follower ensembles for a general two-player game, and considering the possibility of mixed-strategy reputation, still retaining the asymptotic nature of results~\cite{fudenberg1989reputation,fudenberg1992maintaining}.
The ``first-player" advantage, and the entire Stackelberg solution concept, rely on an important assumption: \textit{that the commitment is perfectly revealed to the follower.}
This is usually not the case: in security games, the attacker will usually observe a finite number of deployments of the defender's resource, as opposed to the allocation strategy itself (which is often mixed).
In theoretical models for Bayesian persuasion, the persuader conveys a conditional distribution on her signal given the privately observed state of the world - but what will be practically observed is her history, and thus \textit{realizations} of the signal, not the distribution itself.
In all of these models, the leader establishes her reputation \textit{only partially}, and the manifestation of the revelation is itself random.
It is natural to ask how she should plan to optimally reveal her information under this constraint.

The idea of a robust solution concept in game theory is certainly not new.
The concept of \textit{trembling-hand-perfect-equilibrium}~\cite{selten1975reexamination} explicitly studies how robust mixed-strategy Nash equilibria are to slight perturbations in the mixtures themselves, and a similar concept was proposed for Stackelberg~\cite{van1997games}.
Another solution concept, \textit{quantal-response-equilibrium}~\cite{mckelvey1995quantal}, studies agents that are \textit{boundedly rational}, an orthogonal but important source of uncertainty in response.
In the Stackelberg setting, it was noted that robust commitments \textit{exist} that preserve the Stackelberg guarantee for small enough amounts of noise in the commitment; however, this is still an asymptotic perspective and does not directly help us answer our key computational questions: can we construct a robust commitment efficiently when the game is multi-dimensional, and does the leader want to use the noise to reveal or obfuscate her commitment?

The problem of computing the optimal commitment under \textit{finitely} limited observability corresponds to a robust optimization problem that is, in fact, NP-hard~\cite{an2012security,shieh2012protect}; so directly reasoning about the optimal commitment is not easy.
Whether there exists a polynomial-time approximation scheme for this problem was also unclear.
A duo of papers~\cite{an2012security,shieh2012protect} considered a model of full-fledged observational uncertainty with a Bayesian prior and posterior update based on samples of behavior, and proposed heuristic algorithmic techniques to compute the optimum.
In fact, they also factored for quantal responses using a bounded rationality model~\cite{mckelvey1995quantal}.
This work showed through simulations that there could be a positive return over and above Stackelberg payoff.
In one important piece of analytical work, the problem was also considered for the special case of zero-sum games~\cite{blum2014lazy}, and it was shown that the Stackelberg commitment itself approximated the optimal payoff.
In this result, the extent of approximation actually depends on the amount of observational uncertainty itself - the results we prove for all non-zero-sum games have a similar flavor.

The problem of \textit{communication constraints} in the commitment has also received a lot of interest in the recent algorithmic persuasion literature, but with quantitatively different models for the uncertainty.
Communication constraints on signaling in bilateral trading games~\cite{dughmi2016persuasion} and auction design~\cite{dughmi2014constrained} have been studied from a \textit{compression perspective}, where the leading agent can design the observation channel - while in our model the observation channel is even more constrained to be a finite number of \textit{random} realizations of the mixed commitment.
Further, in many of these settings, the principal is naturally incentivized to reveal the private information and the problem primarily becomes about the communication complexity, and whether the social welfare of the optimal mechanism can be approximated\footnote{These are clearly interesting algorithmic questions in themselves, especially in the case of multiple receivers and private vs public signaling~\cite{dughmi2016algorithmic}, but do not directly address the questions we have raised.}.
In security games, the possibility of the mixed commitment being either fully observed or not observed at all has been considered~\cite{korzhyk2011solving}; as well as different ways of handling the uncertainty, eg: showing that for some security games the Stackelberg and Nash equilibria coincide and observation of commitment does not matter~\cite{korzhyk2011stackelberg}.
Pita et al~\cite{pita2010robust} first proposed a model for the defender (leader) to account for attacker (follower) observational uncertainty by allowing the follower to anchor~\cite{rottenstreich1997unpacking} to a certain extent on its uniform prior.
While they showed significant returns from using their model through extensive experiments, they largely circumvented the algorithmic and analytical challenges by not explicitly considering random samples of defender behavior, thus keeping the attacker response deterministic but shifted.
Our work limits observation in the most natural way for the applications that we consider (i.e. number of samples of leader behavior), and \textit{because the manifestation of the uncertainty is itself random}, our results have distinct and new implications.

\subsection{Our contributions}

Our main contribution is to understand the extent of reputational advantage when interaction is finite, and prescribe approximately optimal commitment rules for this regime.

We study Stackelberg leader-follower games in which a follower obtains a limited number of observations of the leader commitment.
We first prove that in most non-zero-sum games the payoff of the Stackelberg commitment is not robust to even an \textit{infinitesmal amount of} observational uncertainty.
Therefore, the Stackelberg commitment is suboptimal in its payoff\footnote{This property has actually been proved for special examples of Stackelberg games~\cite{van1997games}, but it was unclear whether it holds for all or most games.}.
Next, we propose robust commitment rules for leaders and show that we can approach the Stackelberg payoff as the number of observations increases.
The robust commitment construction involves optimizing a tradeoff between preserving the follower best response and staying close to the ideal Stackelberg commitment, by moving the commitment a little bit into the interior of an appropriate convex polytope~\cite{conitzer2006computing}.
The analysis of payoff of the commitment construction is inspired by interior point convex geometry~\cite{kannan2012random}.
Finally, we show that a possible advantage for the leader from limited observability is only related to follower response mismatch, and show that this advantage is limited. 
Computationally speaking, the corollary is that we are able to approximate the optimal payoff through a simple construction which can be obtained in constant time from computation of the Stackelberg commitment (itself a polynomial-time operation~\cite{conitzer2006computing}).
Philosophically, this result implies that a leader can gain to a very limited extent by misrepresenting her commitment and eliciting a suboptimal response from the follower.
We corroborate our theoretical results with simulations on illustrative examples and random ensembles of security games.

\section{Problem statement}

\subsection{Preliminaries}

We represent a two-player leader-follower game in normal form by the pair of $d \times n$ matrices $(T_1,T_2)$, where $T_1 \in \reals^{d \times n}$ denotes the leader payoff matrix and $T_2 \in \reals^{d \times n}$ denotes the follower payoff matrix.
We denote the leader mixed strategy space by $\Delta_d$ (where $\Delta_k$ for any $k$ represents the $k$-dimensional probability simplex) and the follower mixed strategy space by $\Delta_n$.
From now on, we define an \textit{effective dimension} of a game as a number $m < d$ for which the effective payoff matrices of leader and follower respectively are $A = \begin{bmatrix}
\bolda_1 & \bolda_2 & \ldots & \bolda_n \end{bmatrix} \in \reals^{m \times n}$, $B = \begin{bmatrix} \boldb_1 & \boldb_2 & \ldots & \boldb_n \end{bmatrix} \in \reals^{m \times n}$, and the effective set of leader strategies is given by a \textit{convex polytope} $K \subseteq \Delta_m$\footnote{This definition is important in the context of Stackelberg security games, for which the leader strategy space looks exponential in the number of targets $m$ - but the actual manifestation of all leader strategies is in fact $m$-dimensional. 
In particular, a defender strategy manifests as a distribution over different targets being covered.
}.

We consider a setting of \textit{asymmetric private information} in which the leader knows about the follower preferences (i.e. she knows the matrix $B$) while the follower does not know about the leader preferences (i.e. he possesses no knowledge of the matrix $A$)\footnote{This is an important assumption for the paper, and is in fact used in traditional reputation building frameworks. In future, we will want to better understand situations of repeated interaction where the $\infty$-level leader and $1$-level follower are both learning about one another.}.

With infinite experience, the well-established effect from the follower's point of view is that the leader has established \textit{commitment}, or developed a \textit{reputation}, for playing according to some mixed strategy $\xvec \in K$.
We denote the follower's set of theoretically best \textit{pure-strategy} responses to a mixed strategy commitment $\xvec$ by $\Kspace^*(\xvec) \subseteq [m]$.
Explicitly, we have
\begin{align*}
\Kspace^*(\xvec) := {\arg \max}_{j \in [n]} \inprod{\xvec}{\boldb_j} .
\end{align*}
An important assumption that we make (and that has been made in the classical literature~\cite{conitzer2006computing}) is the follower actually responds with the pure strategy in the set $\Kspace^*(\xvec)$ that is most beneficial to the leader\footnote{The technical reason for this tie-breaking rule is to be able to explicitly define the Stackelberg commitment as an explicit \textit{maximum} - this in itself gives a subtle clue of its fragility.
}.
That is, the follower responds with pure strategy
\begin{align*}
\kstar(\xvec) := {\arg \max}_{j \in \Kspace^*(\xvec)} \inprod{\xvec}{\bolda_j} .
\end{align*}
Then, we also define \textit{best-response regions} as the set of leader commitments that would elicit the pure strategy response $j$ from the follower, i.e. $\Rspace_j := \{\xvec \in K: \kstar(\xvec) = j\}$.

With these definitions, we can define the leader's ideal payoff to be expected with an infinite reputation:
\begin{definition}
A leader with an infinite reputation of playing according to the strategy $\xvec$ should expect payoff
\begin{align*}
f_{\infty}(\xvec) := \inprod{\xvec}{\bolda_{k^*}} .
\end{align*}

Therefore, the leader's \textbf{Stackelberg payoff} is the solution to the program
\begin{align*}
f^*_{\infty} := \max_{\xvec \in \Delta_m} f_{\infty}(\xvec) .
\end{align*}

The argmax of this program is denoted as the \textbf{Stackelberg commitment} $\xvec^*_{\infty}$.
Further, we denote the best response faced in Stackelberg equilibrium by $\kstar := \kstar(\xstar_{\infty})$.
\end{definition}

It is clear that the Stackelberg commitment is optimal for the leader under two conditions: \textit{the leader is 100\% known to be committed to a fixed strategy}, and \textit{the follower knows exactly the leader's committed-to strategy}.
For a finite number of interactions, neither is true.

\subsection{Observational uncertainty with established commitment}\label{sec:obsmodel}

Even assuming that there is a shared belief in commitment, there is uncertainty.
In particular, with a finite number of plays, the follower does not know the exact strategy that the leader has committed to, and only has an estimate.

Consider the situation where a leader can only reveal its commitment $\xvec$ through $N$ \textit{pure strategy plays} $I_1,I_2,\ldots,I_N \text{ i.i.d } \sim \xvec$.
The commitment is known (to both leader and followers) to come from a set of mixed strategies $\Xspace \subseteq K$.
We denote the maximum likelihood estimate of the leader's mixed strategy, as seen by the follower, by $\Xhat_N$.
It is reasonable to expect, under certainty of commitment, that a ``rational"\footnote{Rational is in quotes because the follower is not necessarily using expected-utility theory (although there is an expected-utility-maximization interpetation to this estimate if the mixed strategy were uniform drawn from $\Xspace$).} follower would best-respond to $\Xhat_N$, i.e. play the pure strategy
\begin{align}\label{eq:fNbestresponse}
\kstar(\Xhat_N) .
\end{align}
We can express the expected leader payoff under this learning rule.
\begin{definition}
A leader which will have $N$ plays according to the hidden strategy $\xvec$ can expect payoff in the $N$th play of
\begin{align*}
f_N(\xvec) := \EE\left[\inprod{\xvec}{\bolda_{\kstar(\Xhat_N)}}\right] 
\end{align*}

against a follower that plays according to (\ref{eq:fNbestresponse}).
The maximal payoff a leader can expect is 
\begin{align*}
\fstar_N := \max_{\xvec \in \Delta_m} f_N(\xvec)
\end{align*}

and it acquires this payoff by playing the argmax strategy $\xvec^*_N$.
\end{definition}

Ideally, we want to understand how close $\fstar_N$ is to $\fstar_{\infty}$, and also how close $\xstar_N$ is to $\xstar_{\infty}$.
An answer to the former question would tell us how observational uncertainty impacts the first-player advantage.
An answer to the latter question would shed light on whether the best course of action deviates significantly from Stackelberg commitment.
We are also interested in algorithmic techniques for \textit{approximately computing} the quantity $\fstar_N$, as doing so exactly would involve solving a non-convex optimization problem.

\section{Main Results}

\subsection{Robustness of Stackelberg commitment to observational uncertainty}

A natural first question is whether the Stackelberg commitment, which is clearly optimal if the game were being played infinitely (or equivalently, if the leader had infinite commitment power and exact public commitment), is also suitable for finite play.
In particular, we might be interested in evaluating whether we can do better than the baseline Stackelberg performance $\fstar$.
We show through a few paradigmatic examples that the answer can vary.

\multifigureexterior{fig:games}{Illustration of examples of zero-sum game and non-zero-sum games in the form of normal form tables and ideal leader payoff function $f_{\infty}(.)$. $p$ denotes the probability that the leader will play strategy $1$, and fully describes leader mixed commitment for these $2 \times n$ games.}{
\subfigl{0.33\textwidth}{fig:zerosum1}{$2 \times 3$ zero-sum game.}{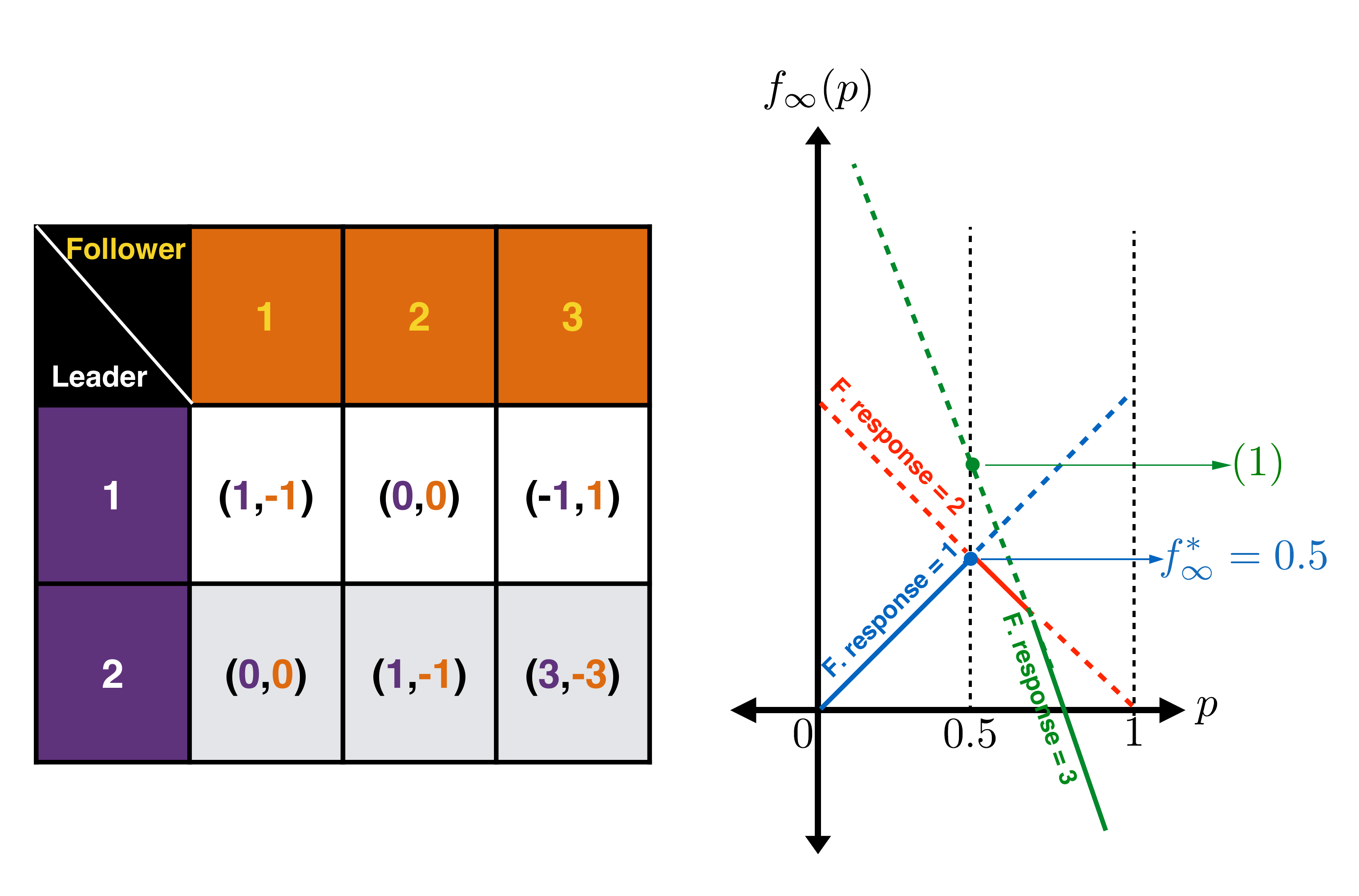}  
\subfigl{0.33\textwidth}{fig:nonzerosumbad1}{$2 \times 2$ non-zero-sum game.}{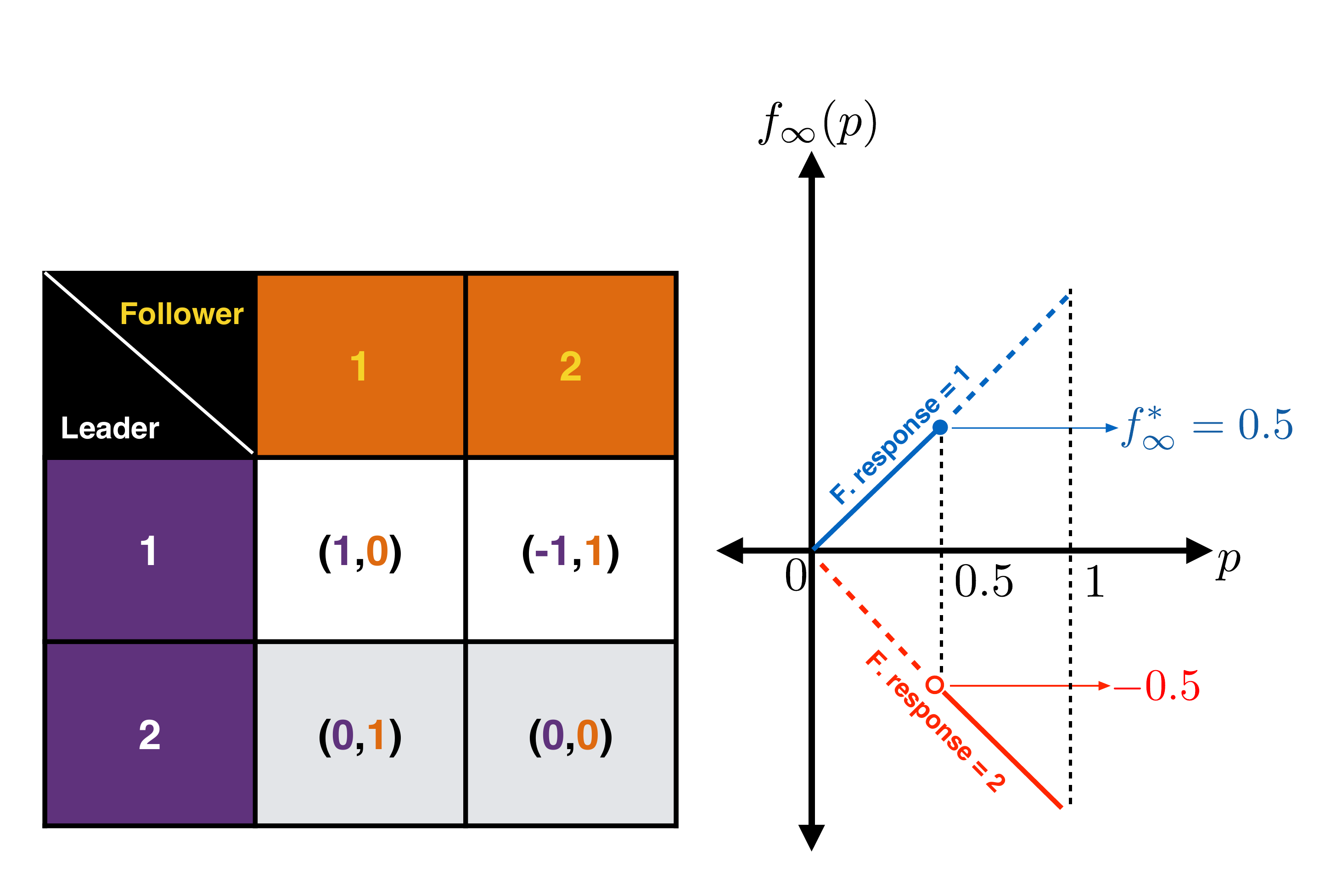}
\subfigl{0.33\textwidth}{fig:nonzerosum1}{$2 \times 3$ non-zero-sum game.}{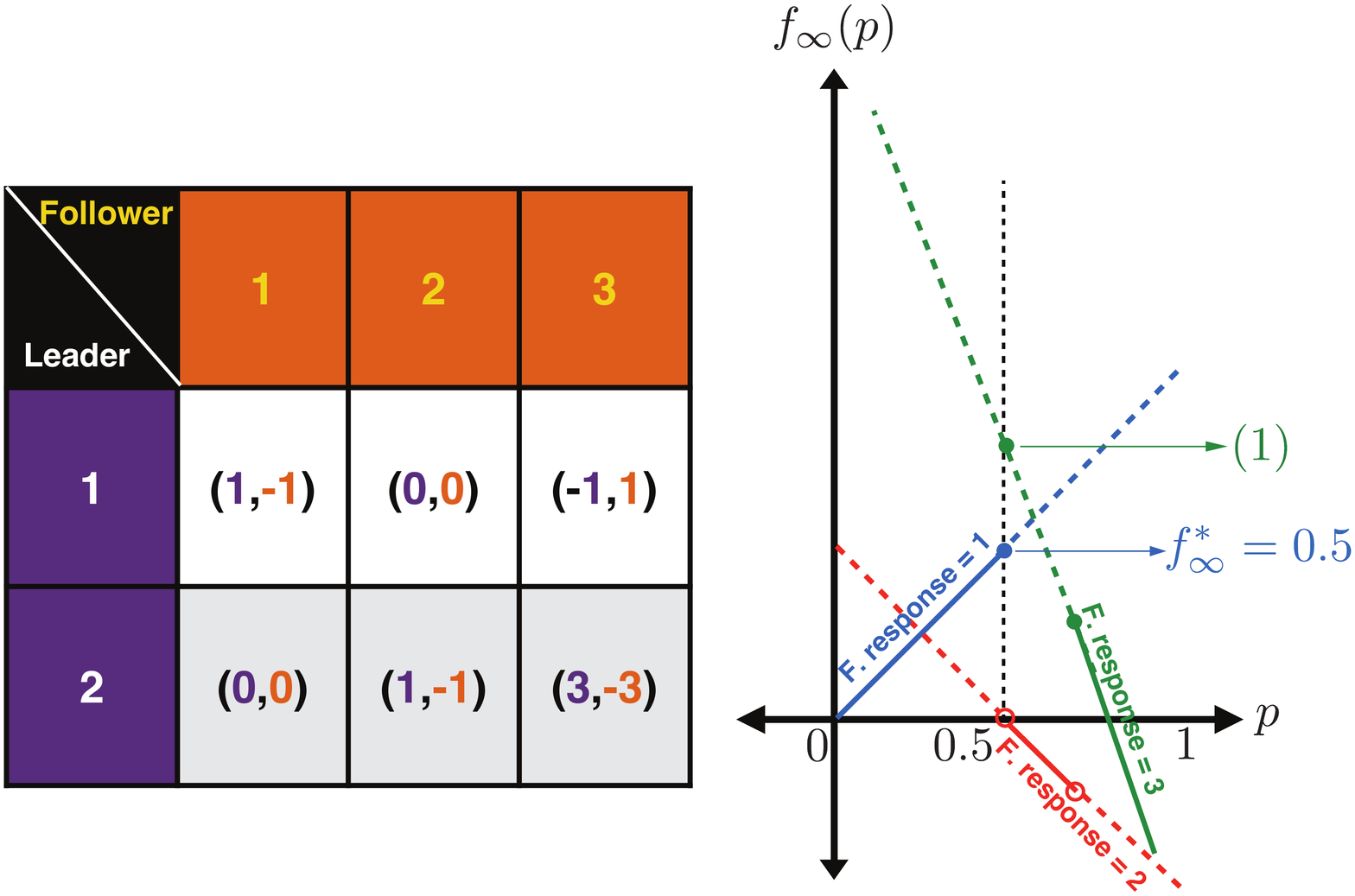}
}

\begin{example}\label{eg:zerosum}

\quickfigure{fig:zerosum2}{Semilog plot of extent of advantage over Stackelberg payoff as a function of $N$ in the $2 \times 3$ zero-sum game depicted in Figure~\ref{fig:zerosum1}.}{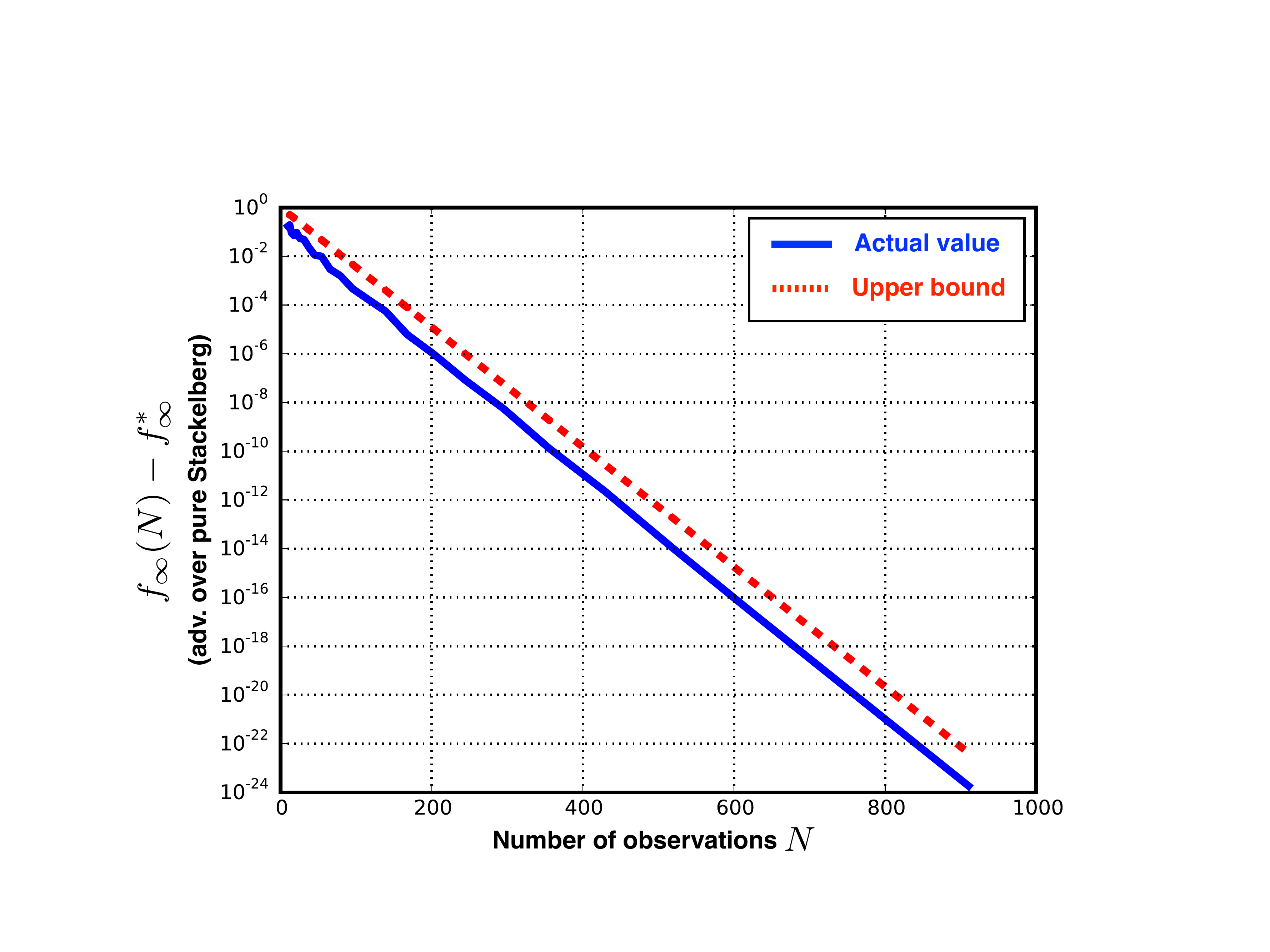}{0.33\textwidth}


We consider a $2 \times 3$ \textit{zero-sum game}, represented in normal form in Figure~\ref{fig:zerosum1}, in which we can express the leader strategy according to the probability $p$ with which she will pick strategy $1$, and leader payoff is as follows:
\begin{align*}
\begin{cases}
f(p;1) := p \text{ if follower best responds with strategy 1 } \\
f(p;2) := 1 - p \text{ if follower best responds with strategy 2 } \\
f(p;3) := 3 - 4p \text{ if follower best responds with strategy 3 } \\
\end{cases}
\end{align*}

Since the game is zero-sum, the follower responds in a way that is worst-case for the leader.
This means that we can express the leader payoff as
\begin{align*}
f_{\infty}(p) = \min\{f(p;1),f(p;2),f(p;3)\} .
\end{align*} 

This leader payoff structure is depicted in Figure~\ref{fig:zerosum1}.
Therefore, we can express the Stackelberg payoff as 
\begin{align*}
\fstar_{\infty} = \max_{p \in [0,1]} f_{\infty}(p) = 1/2,
\end{align*} 
attained at $\pstar_{\infty} = 1/2$.
We wish to evaluate $f_N(\pstar_{\infty})$.
It was noted~\cite{blum2014lazy} that $f_N(\pstar_{\infty}) \geq \fstar_{\infty}(\pstar_{\infty})$ by the minimax theorem, but not always clear whether strict inequality would hold (that is, if observational uncertainty gives a strict advantage).
For this example, we can actually get a sizeable improvement! 
To see this, look at the simple example of $N = 1$.
Denoting $\Phat = \frac{1}{N}\sum_{j=1}^N \Ind[I_j = 1]$, we have
\begin{align*}
f_1(1/2) &= \Pr[\Phat_1 = 0].f(1/2;1) + \Pr[\Phat_1 = 1].f(1/2;3) \\
&= 1/2 \times 1/2 + 1/2 \times 1 = 3/4 . \\
\end{align*}
The semilog plot in Figure~\ref{fig:zerosum2} shows that this improvement persists for larger values of $N$, although the extent of improvement decreases \textit{exponentially} with $N$.
We can show that
\begin{align*}
f_N(1/2) - f^*_{\infty} &= 1/2 \Pr[\Phat_N > 2/3] \\
&= 1/2 \Pr[\Phat_N - 1/2 > 1/6 ] \\
&\leq \exp\{-N\kull{2/3}{1/2}\}
\end{align*}

where $\kull{.}{.}$ denotes the Kullback-Leibler divergence, and the last inequality is due to Sanov's theorem~\cite{csiszar2011information}.

This shows analytically that the advantage does indeed decrease exponentially with $N$.
Naturally, this is because we see a decrease in the stochasticity that elicits the more favorable follower response with action $3$ with the number of observations $N$.
\end{example}

Example~\ref{eg:zerosum} showed us how Stackelberg commitment power could be increased by stochastically eliciting more favorable responses.
We now see an example illustrating that the commitment power can disappear completely.

\begin{example}\label{eg:nonzerosumbad}

\multifigureexterior{fig:nonzerosumbad}{Example of the $2 \times 2$ non-zero-sum game depicted in Figure~\ref{fig:nonzerosumbad1}, for which observational uncertainty is always undesirable.}{
\subfigl{0.33\textwidth}{fig:nonzerosumbad2}{Plot depicting the performance of the sequence of robust commitments $\{\xvec_N\}_{N \geq 1}$ and the Stackelberg commitment $\xstar_{\infty}$ as a function of $N$. The benchmark for comparison is idealized Stackelberg payoff $\fstar_{\infty}$.}{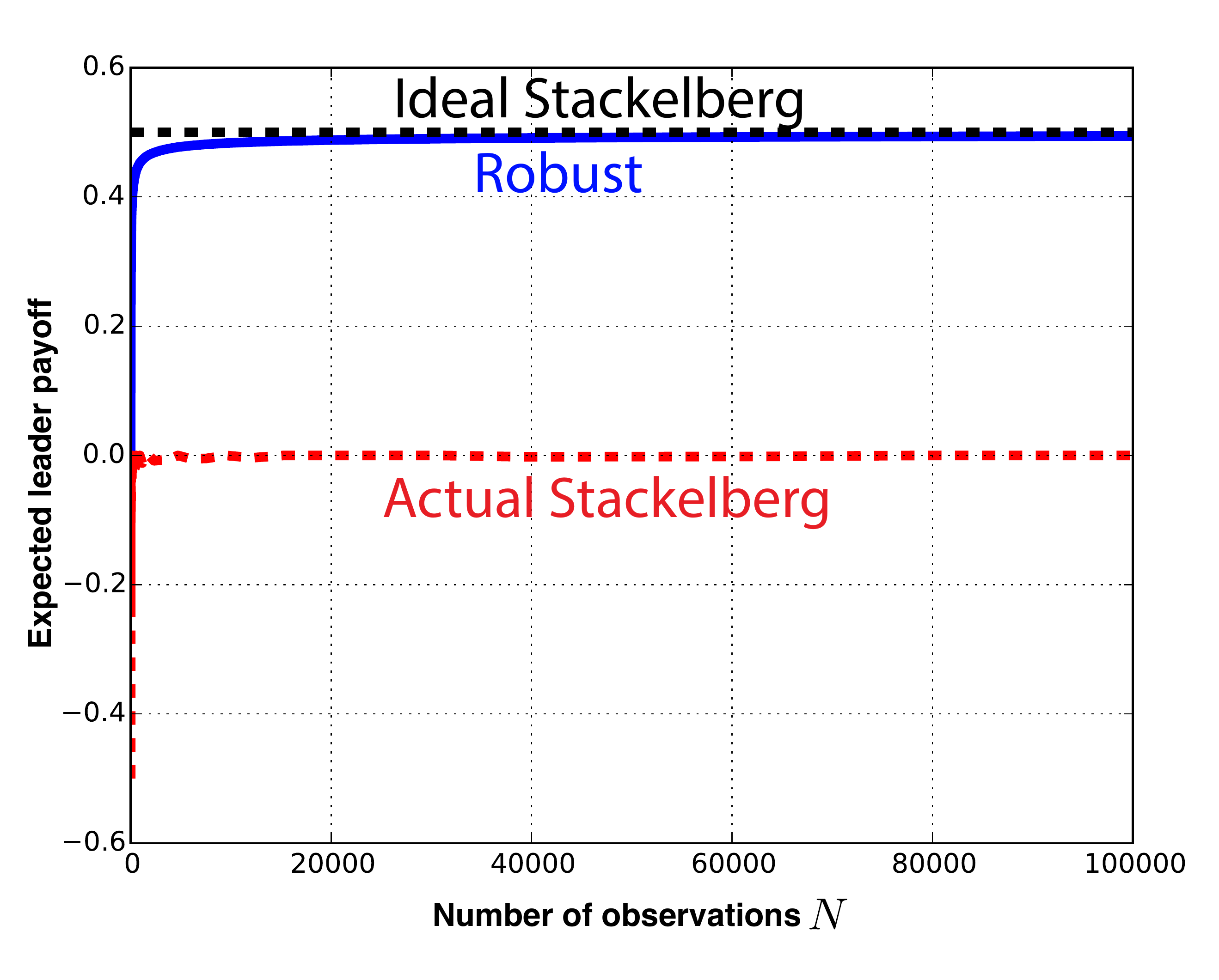}
\subfigl{0.33\textwidth}{fig:nonzerosumbad3}{Plot showing the performance of sequence of robust commitments $\{\xvec_N\}_{N \geq 1}$ as compared to the optimum performance $\fstar_N$ (brute-forced).}{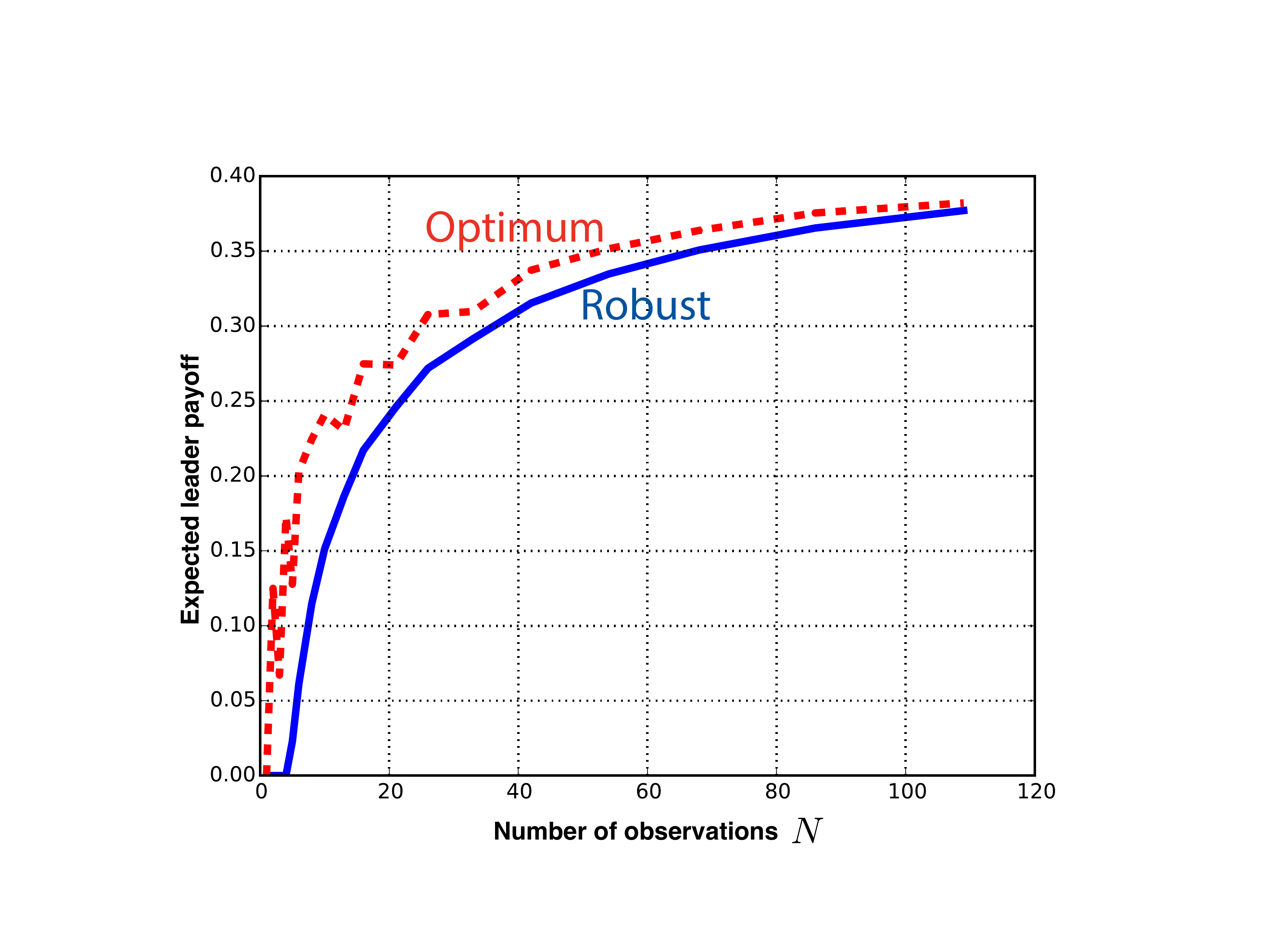}}

We consider a $2 \times 2$ \textit{non-zero-sum game}, represented in normal form and leader payoff structure in Figure~\ref{fig:nonzerosumbad1}.
Explicitly, the ideal leader payoff function is
\begin{align*}
f_{\infty}(p) = \begin{cases}
p \text{ if } p \leq 1/2 \\
-p \text{ if } p > 1/2 .
\end{cases}
\end{align*}

This is essentially the example reproduced in~\cite{blum2014lazy}, which we repeat for storytelling value.
Notice that $\fstar_{\infty} = 1/2, \pstar_{\infty} = 1/2$, but the advantage evaporates with \textit{observational uncertainty}.
For any finite $N$, we have
\begin{align*}
f_N(1/2) &= \Pr[\Phat_N \leq 1/2](1/2) + \Pr[\Phat_N > 1/2](-1/2) \\
&= 1/2 \times 1/2 - 1/2 \times 1/2 = 0 .
\end{align*}

Remarkably, this implies that $\fstar_{\infty} - f_N(\pstar_{\infty}) = 1/2$ and so $\lim_{N \to \infty} \fstar_{\infty} - f_N(\pstar_{\infty}) \neq 0$!
This is clearly a very negative result for the robustness of Stackelberg commitment, and as a very pragmatic matter tells us that the idealized Stackelberg commitment $\pstar_{\infty}$ is far from ideal in finite-observation settings.
This example shows us a case where stochasticity in follower response is \textit{not desired}, principally because of the discontinuity in the leader payoff function at $\pstar_{\infty}$.
\end{example}

Example~\ref{eg:nonzerosumbad} displayed to the fullest the significant disadvantage of observational uncertainty.
The game considered was special in that there was no potential for limited-observation gain, while in the game presented in Example~\ref{eg:zerosum} there was only potential for limited-observational gain. 
What could happen in general? 
Our next and final example provides an illustration.

\begin{example}\label{eg:nonzerosum}

\multifigureexterior{fig:nonzerosum}{Example of the $2 \times 3$ non-zero-sum game depicted in Figure~\ref{fig:nonzerosum1}, in which observational uncertainty could either help or hurt the leader.}{
\subfigl{0.33\textwidth}{fig:nonzerosum2}{Plot of the extent of (dis)advantage over Stackelberg payoff as a function of $N$.}{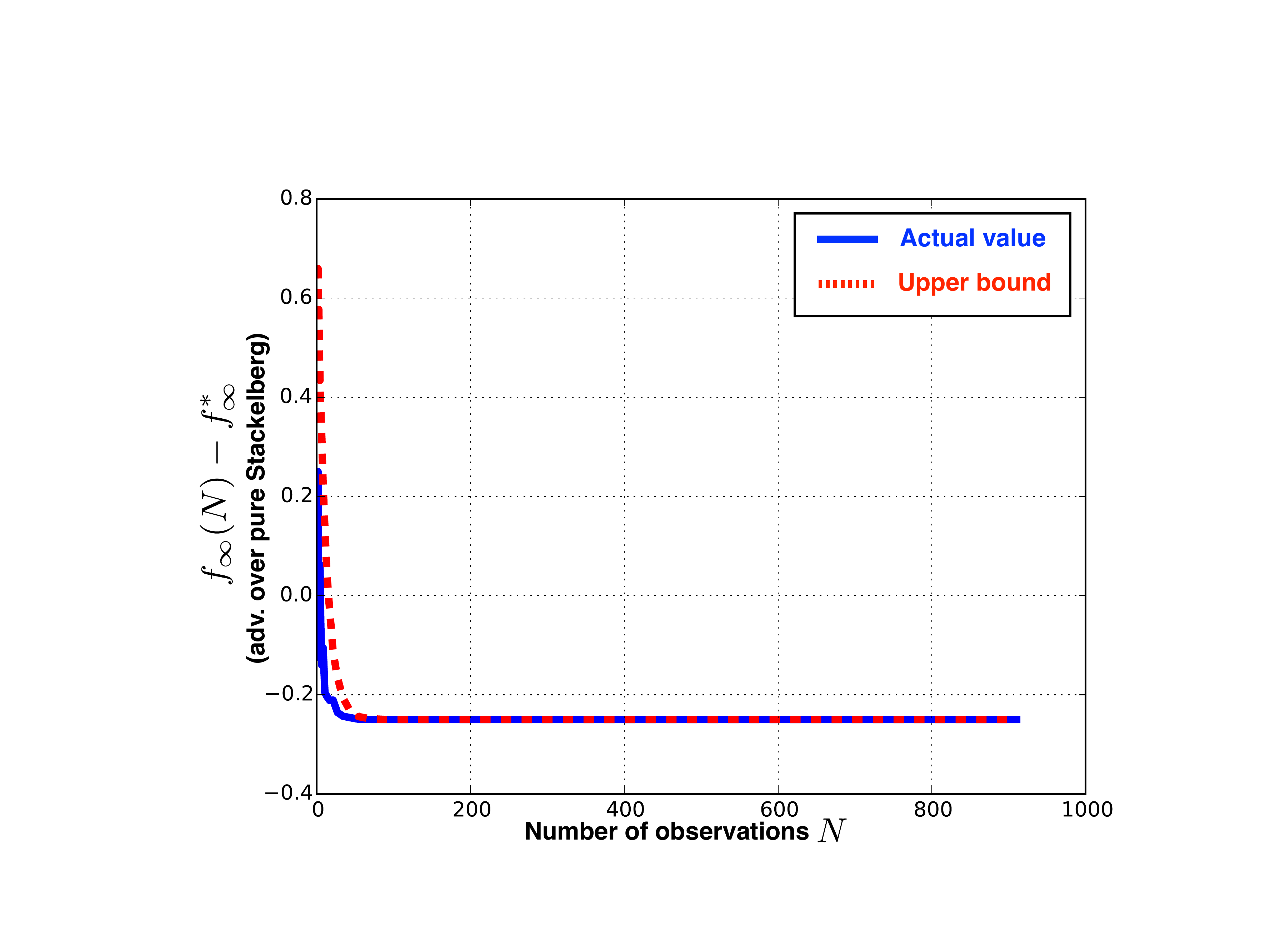}
\subfigl{0.33\textwidth}{fig:nonzerosum3}{Plot depicting the performance of the sequence of robust commitments $\{\xvec_N\}_{N \geq 1}$ and the Stackelberg commitment $\xstar_{\infty}$ as a function of $N$. The benchmark for comparison is idealized Stackelberg payoff $\fstar_{\infty}$.}{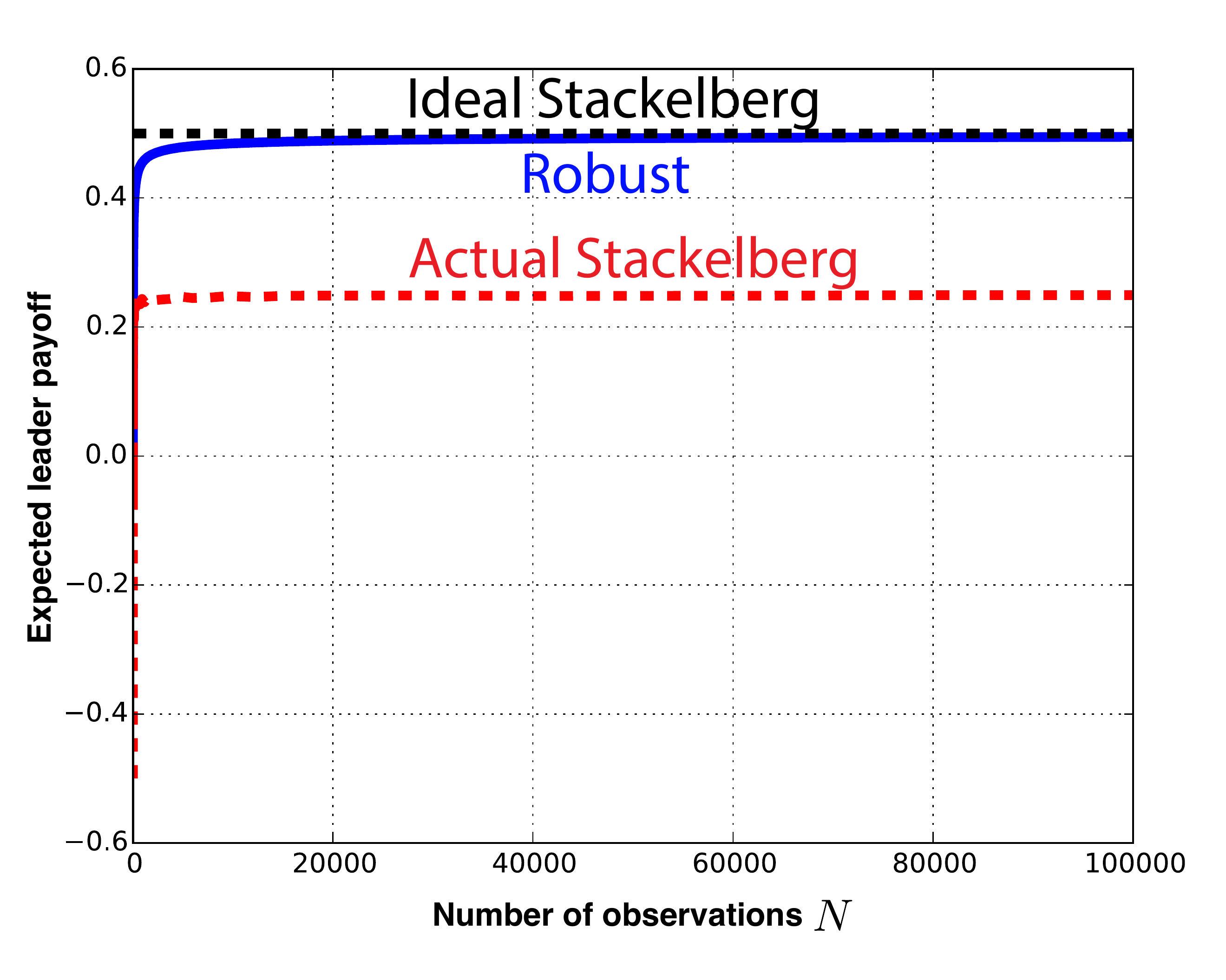}
\subfigl{0.33\textwidth}{fig:nonzerosum4}{Plot showing the performance of sequence of robust commitments $\{\xvec_N\}_{N \geq 1}$ as compared to the optimum performance $\fstar_N$ (brute-forced).}{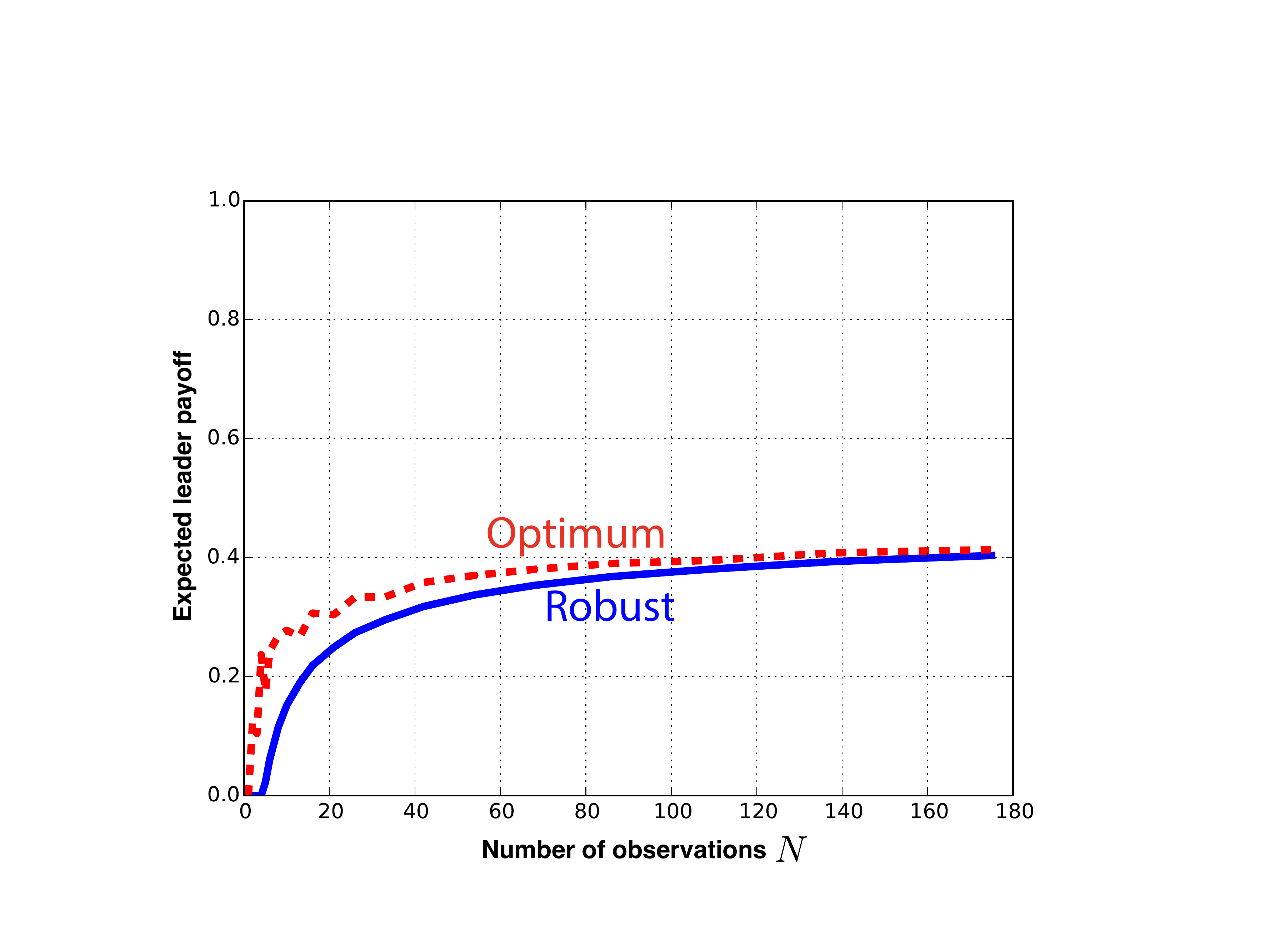}}

Our final example considers a $2 \times 3$ non-zero-sum game, whose normal form and leader payoff structure are depicted in Figure~\ref{fig:nonzerosum1}.
The ideal leader payoff function is 
\begin{align*}
f_{\infty}(p) = \begin{cases}
p \text{ if } p \leq 1/2 \\
1/2 - p \text { if } 1/2 < p \leq 5/7 \\
3 - 4p \text{ if } p > 5/7 .
\end{cases}
\end{align*}

As in the other examples, $\fstar_{\infty} = 1/2, \pstar_{\infty} = 1/2$.
Notice that this example captures both positive and negative effects of stochasticity in response.
On one hand, follower response $2$ is highly undesirable (a la Example~\ref{eg:nonzerosumbad}) but follower response $3$ is highly desirable (a la Example~\ref{eg:zerosum}).
What is the net effect?
We have
\begin{align*}
f_N(1/2) &= \Pr[\Phat_N \leq 1/2](1/2) + \Pr[1/2 < \Phat_N < 5/7](0) + \Pr[\Phat_N \geq 5/7](1) \\
&= (1/2)(1/2) + \Pr[\Phat_N \geq 5/7](1) \\
&\leq 1/4 + 1/2\exp\{-N\kull{5/7}{1/2}\} .
\end{align*}

A quick calculation thus tells us that $f_N(\pstar_{\infty}) <= \fstar_{\infty}$ if $N \geq 8$, showing that Stackelberg in fact has poor robustness for this example.
Intuitively, the probability of the ``bad" stochastic event remains constant while the probability of the ``good" stochastic event decreases exponentially with $N$.
Even more damningly, we see that $\lim_{N \to \infty} \fstar_{\infty} - f_N(\pstar_{\infty}) \geq \lim_{N \to \infty} 1/4 - 1/2\exp\{-N\kull{5/7}{1/2}  = 1/4$, again showing that the Stackelberg commitment is far from ideal.
We can see the dramatic decay of leader advantage over and above Stackelberg, and ensuing disadvantage even for a very small number of observations, in Figure~\ref{fig:nonzerosumbad2}.
\end{example}

While the three examples detailed above provide differing conclusions, there are some common threads.
For one, in all the examples it is the case that committing to the Stackelberg mixture $\xstar_{\infty}$ can result in the follower being agnostic between more than one response.
Only one of these responses, the pure strategy $\kstar$, is desirable for the leader.
A very slight misperception in the estimation of the true value $\xstar_{\infty}$ can therefore lead to a different, worse-than-expected response and this misperception happens with a sizeable, non-vanishing probability.
On the flipside, a different response could also lead to better-than-expected payoff, raising the potential for a gain over and above $\fstar$.
However, these \textit{better-than-expected responses} cannot share a boundary with the Stackelberg commitment, and we will see that the probability of eliciting them decreases exponentially with $N$.
The net effect is that the Stackelberg commitment is, most often, not robust -- and critically, this is even the case for small amounts of uncertainty.

Our first result is a formal statement of instability of Stackelberg commitments for a general $2 \times n$ game.
We denote the leader probability of playing strategy $1$ by $p \in [0,1]$, and the Stackelberg commitment's probability of playing strategy $1$ by $\pstar_{\infty}$.

Furthermore, let $\phi(t)$ denote the CDF of the standard normal distribution $\NORMAL(0,1)$.
We are now ready to state the result.

\begin{theorem}\label{thm:stackelbergnotrobust}
For any $2 \times n$ leader-follower game in which $\pstar_{\infty} \in (0,1)$ and $f_{\infty}(p)$ discontinuous at $p = \pstar_{\infty}$, we have
\begin{align}\label{eq:stackelbergnotrobust}
f_N(\pstar_{\infty}) \leq \fstar_{\infty} - C\left(\phi(\sqrt{N}C') - \frac{1}{2}- \frac{C'}{\sqrt{N}}\right) + \exp\{-NC''^2\}
\end{align}

where $C,C',C''$ are strictly positive constants depending on the parameters of the game.
This directly implies the following:
\begin{enumerate}
\item For some $N_0 > 0$, we have $f_N(\pstar_{\infty}) < \fstar_{\infty}$ for all $N > N_0$.
\item We have $\lim_{N \to \infty} f_N(\pstar_{\infty}) < \fstar_{\infty}$.
\end{enumerate}
\end{theorem}

The proof of Theorem~\ref{thm:stackelbergnotrobust} is contained in Section~\ref{sec:prop1proof}.
The technical  ingredients in the proof are the Berry-Esseen theorem~\cite{berry1941accuracy,esseen1942liapounoff}, used to show that the detrimental alternate responses on the Stackelberg boundary are non-vanishingly likely -- and the Hoeffding bound, used to tail bound the probability of potentially beneficial alternate responses not on the boundary\footnote{It is worth noting that a similar argument as presented here could be extended to a general $m \times n$ game, using iid random vectors instead of random variables and considering a demarcation into best-response regions as illustrated in Figure~\ref{fig:prop1illustration}. We restrict attention to the $2 \times n$ case for ease of exposition.}

For non-robustness of Stackelberg commitment to hold, the two critical conditions for the game are that there is a discontinuity at the Stackelberg boundary, and that the Stackelberg commitment is mixed. For a zero-sum game, the first condition does not hold and the Stackelberg commitment stays robust as we saw in Example~\ref{eg:zerosum}.

The theorem directly implies that the ideal Stackelberg payoff is only obtained for the exact case of $N = \infty$ (when the commitment is perfectly observed), and that for any value of $N < \infty$ there is a \textit{non-vanishing reduction} in payoff.
In the simulations in Section~\ref{sec:simulations}, we will see that this gap is empirically significant.

\subsection{Robust commitments achieving close-to-Stackelberg performance}

The surprising message of Theorem~\ref{thm:stackelbergnotrobust} is that, in general, the Stackelberg commitment $\xstar$ is undesirable.
The commitment $\xstar$ is pushed to the \textit{extreme point} of the best-response-region $\Rspace_{\kstar}$ to ensure optimality under idealized conditions; and this is precisely what makes it sub-optimal under uncertainty.
What if we could move our commitment a little bit into the interior of the region $\Rspace_{\kstar}$ instead, such that we can get a \textit{high-probability-guarantee} on eliciting the expected response, while staying sufficiently close to the idealized optimum?
Our next result quantifies the ensuing tradeoff and shows that we can cleverly construct the commitment to approximate Stackelberg performance.

\begin{theorem}\label{thm:mtimesnobs}
Let the best-response polytope $\Rspace_{\kstar}$ be non-empty in $\reals^{m-1}$.
Then, provided that the number of samples $N = \widetilde{\Oh}(m)$, we can construct commitment $\xvec_{N,p}$ for every $0 < p < 1/2$ such that 
\begin{align}
\fstar_{\infty} - f_N(\xvec_N) &= \widetilde{\Oh}\Big( \left(\frac{m}{N}\right)^p + e^{-\omega(1) \cdot N^{1 - 2p}}  \Big) .
\end{align}

Furthermore, these constructions are computable in $\Oh(1)$ time with knowledge of the Stackelberg commitment $\xstar_{\infty}$.
(The $\widetilde{\Oh}(\cdot)$ contains constant factors that depend on both the local and global geometry of the best-response-region $\Rspace_{\kstar}$. For a fully formal statement that includes these factors, see Lemma~\ref{lem:mainthmformal}.)
\end{theorem}

The full proof of Theorem~\ref{thm:mtimesnobs}, deferred to Appendix~\ref{sec:thm1proof}, involves some technical steps to achieve as good as possible a scaling in $N$.
The caveat of Theorem~\ref{thm:mtimesnobs} is that commitment power can be robustly exploited in this way only if there are enough observations of the commitment. 
One obvious requirement is that the best-response-region $\Rspace_{\kstar}$ needs to be non-empty in $\reals^{m-1}$.
Second, the number of observations $N$ needs to be greater than the \textit{effective dimension} of the game for the leader, $m$.
This is a natural requirement to ensure that the follower has learned at least a meaningful estimate of the commitment.
Third, the ``constant" factors in Theorem~\ref{thm:mtimesnobs} actually reflect properties about both the local and global geometry of the polytope; see Appendix~\ref{sec:thm1proof} for more details.
Intuitively, geometric properties that lead to undesirable scaling in the constant factors in the robustness guarantee are listed below:
\begin{enumerate}
\item The Stackelberg commitment being a ``pointy" vertex: this can lead to a commitment being far away from the boundary in certain directions, but closer in others, making it more likely for a different response to be elicited.
\item Local constraints being very different from global constraints, which implies that commitments too far in the interior of the local feasibility set will no longer satisfy all the constraints of the best-response-region.
\end{enumerate}

Even with these caveats, Theorem~\ref{thm:mtimesnobs} provides an attractive general framework for constructing robust commitments by making a natural connection to interior-point methods in optimization\footnote{Noting that interior point methods are provably polynomial-time algorithms to solve LPs, it is plausible to think that in fact, stopping the interior point method appropriately early would also give us a robustness guarantee - which would imply that finding optimal \textit{robust} commitments is even easier than finding optimal commitments!}.
We observe significant empirical benefit from the constructions in the simulations in Section~\ref{sec:simulations}.

We also mention a couple of special cases of leader-follower games for which the robust commitment constructions of Theorem~\ref{thm:mtimesnobs} are not required; in fact, it is simply optimal to play Stackelberg.

\begin{remark}
For games in which the mixed-strategy Stackelberg equilibrium coincides with a pure strategy, the follower's best response is always as expected regardless of the number of observations.
There is no tradeoff and it is simply optimal to play Stackelberg even under observational uncertainty.
\end{remark}

\begin{remark}
For the zero-sum case, it was observed~\cite{blum2014lazy} that a Stackelberg commitment is made assuming that the follower will respond in the worst case.
If there is observational uncertainty, the follower can only respond in a way that yields payoff for the leader that is better than expected.
This results in an expected payoff greater than or equal to the Stackelberg payoff $\fstar_{\infty}$, and it simply makes sense to stick with the Stackelberg commitment $\xstar_{\infty}$.
As we have seen, this logic does not hold up for non-zero-sum games because different responses can lead to worse-than-expected payoff.
One way of thinking of this is that the function $f_{\infty}(.)$ can generally be discontinuous in $\xvec$ for a non-zero-sum game, but is always continuous for the special case of zero-sum.
\end{remark}

\subsection{Approximating the maximum possible payoff}

So far, we have considered the limited-observability problem and shown that the Stackelberg commitment $\xstar_{\infty}$ is not a suitable choice.
We have constructed robust commitments that come close to idealized Stackelberg payoff $\fstar_{\infty}$ and shown that the guarantee fundamentally depends on the number of observations scaling with the effective dimension of the game.
Now, we turn to the question of whether we can approximate $\fstar_N$, the actual optimum of the program.
Note that since the problem is in general non-convex in $\xvec$, it is NP-hard to exactly compute.

Rather than the traditional approach of constructing a polynomial-time-approximation-algorithm, our approach is approximation-theoretic\footnote{In other words, the extent of approximation is measured by the \textit{number of samples} as opposed to the runtime of an algorithm. 
This is very much the flavor of previously-obtained results on Stackelberg zero-sum security games~\cite{blum2014lazy}.}.
We first show that in the large-sample case, we cannot do much better than the actual Stackelberg payoff $\fstar_{\infty}$; informally speaking, our ability to fool the follower into responding \textit{strictly-better-than-expected} is limited.
Combining this with the robust commitment construction of Theorem~\ref{thm:mtimesnobs}, we obtain an approximation to the optimum payoff.

The main result of this section is stated below.
\begin{theorem}\label{thm:approximation}
We have
\begin{align*}
f^*_N &\leq f^*_{\infty} + Cn\sqrt{\frac{m}{N}} .
\end{align*}

for some constant $C > 0$ depending on the parameters of the game $(A,B)$.
\end{theorem}

As a corollary the commitment construction defined in Theorem~\ref{thm:mtimesnobs} provides a $\widetilde{\Oh}(\sqrt{\frac{1}{N}})$-additive-approximation algorithm to $\fstar_N$.
The proof of Theorem~\ref{thm:approximation} is provided in the appendix.

Intellectually, Theorem~\ref{thm:approximation} tells us that the robust commitments are essentially optimal.
The practical benefit that Theorem~\ref{thm:approximation} affords us is that we now have an approximation to the optimum payoff the leader could possibly obtain, which can be computed in constant time after computing the Stackelberg equilibrium, which itself is polynomial time~\cite{conitzer2006computing}.
This is because the robust commitment is obtained by first computing Stackelberg equilibrium $\xstar_{\infty}$, and then deviating away from $\xstar_{\infty}$ in the magnitude and direction specified.
We will now study the empirical benefits of our robust commitment constructions.

\section{Simulations}\label{sec:simulations}



\subsection{Example $2 \times 2$ and $2 \times 3$ games}

First, we return to the non-zero-sum games described in Examples~\ref{eg:nonzerosumbad} and ~\ref{eg:nonzerosum}.
These were $2 \times 2$ and $2 \times 3$ games respectively, and the Stackelberg commitment was non-robust for both games.
Now, armed with the results in Theorem~\ref{thm:mtimesnobs}, we can employ our robust commitment constructions and study their performance.
To construct our robust commitments, we first computed the Stackelberg commitment using the LP solver in scipy (scipy.optimize.linprog), and then used the construction in Theorem~\ref{thm:mtimesnobs}.

Figures~\ref{fig:nonzerosumbad2} and ~\ref{fig:nonzerosum3} compares the expected payoff obtained by our robust commitment construction scheme $\{\xvec_N\}_{N \geq 1}$ for different numbers of samples $N$, and for the games described in Examples~\ref{eg:nonzerosumbad} and ~\ref{eg:nonzerosum} respectively.
The benchmark with respect to which we measure this expected payoff is the Stackelberg payoff $\fstar_{\infty}$ (obtained by Stackelberg commitment under \textit{infinite observability} and \textit{tie-breakability in favor of the leader}).
We also observe a significant gap between the payoffs obtained by these robust commitment constructions and the payoff obtained if we used the Stackelberg commitment $\xstar_{\infty}$.
We showed in theory that there is significant benefit for choosing the commitment to factor in such observational uncertainty, and we can now see it in practice.

Furthermore, for the case of $2$ leader actions we were able to brute-force the maximum possible obtainable payoff\footnote{First we used scipy.optimize.brute with an appropriate grid size to initialize, and then ran a gradient descent at that initialization point. This was feasible for the case of $2$ pure strategies.} $\fstar_N$, and compare the value to the robust commitment payoff.
This comparison is particularly valuable for smaller values of $N$, as shown in Figures~\ref{fig:nonzerosumbad3} and~\ref{fig:nonzerosum4}.
We notice that the values are much closer even than our theory would have predicted, and even for small values of $N$.
Thus, our constructions have significant practical benefit as well: we are able to get close to the optimum while drastically reducing the required computation (to just solving $n$ LPs!).

Since these examples involved $2 \times n$ games, the commitment construction became trivial (i.e. only one direction to move along) -- next, we test our commitment constructions for $m \times m$ security games.
Instead of looking at specific examples, we now look at a random ensemble to see what behavior ensues.

\subsection{Random security games}

\quickfigure{fig:randomsecuritygame0}{Illustration of random ensemble of $5 \times 5$ security game.}{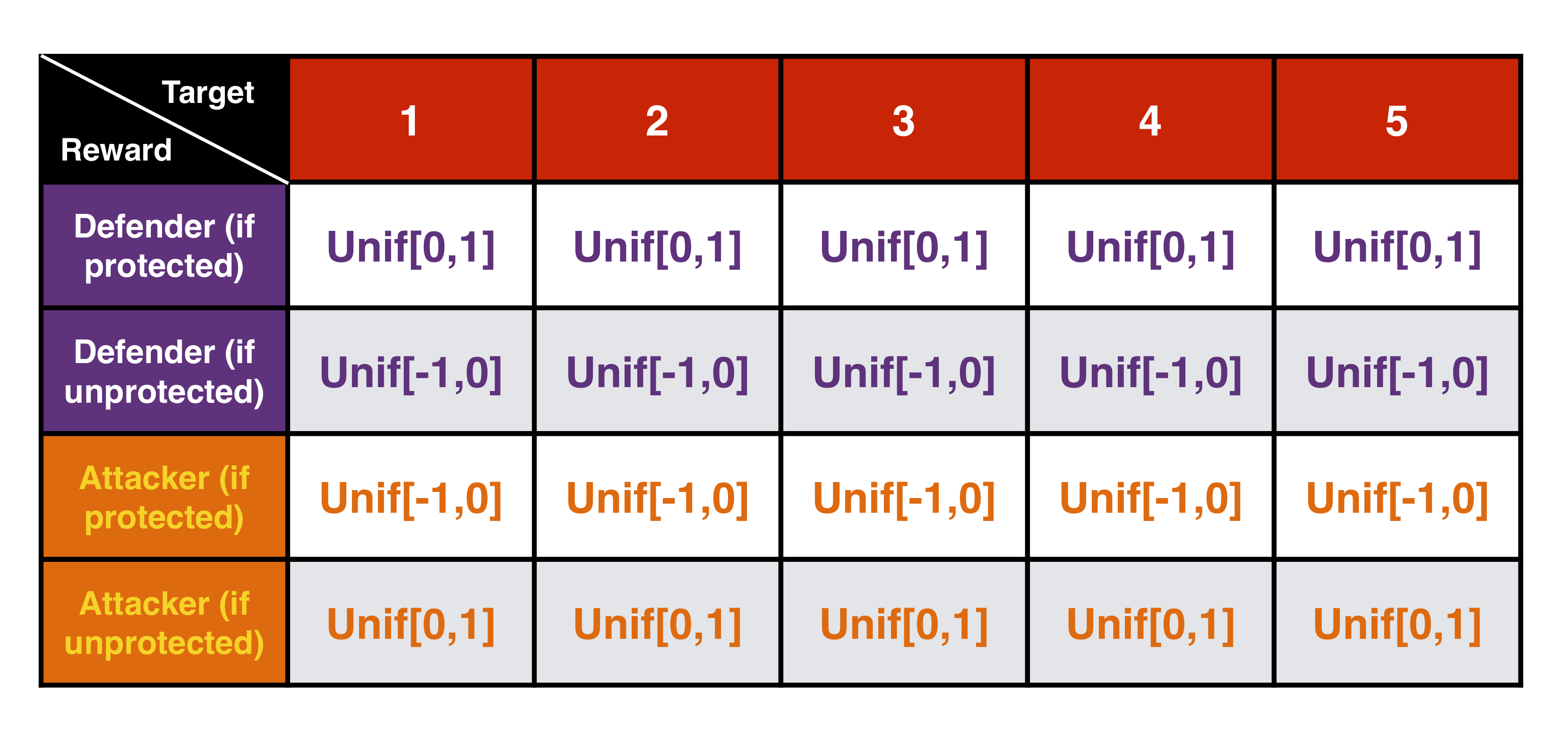}{0.5\textwidth}

\multifigureexterior{fig:randomsecuritygame}{Illustration of performance of robust commitments and Stackelberg commitment in random $5 \times 5$ Stackelberg security games for a finite number of observations of defender commitment.}{
\subfigl{0.33\textwidth}{fig:randomsecuritygame1}{Plot of expected defender payoff when defender uses robust commitments -- compared to Stackelberg commitment as well as idealized Stackelberg payoff.}{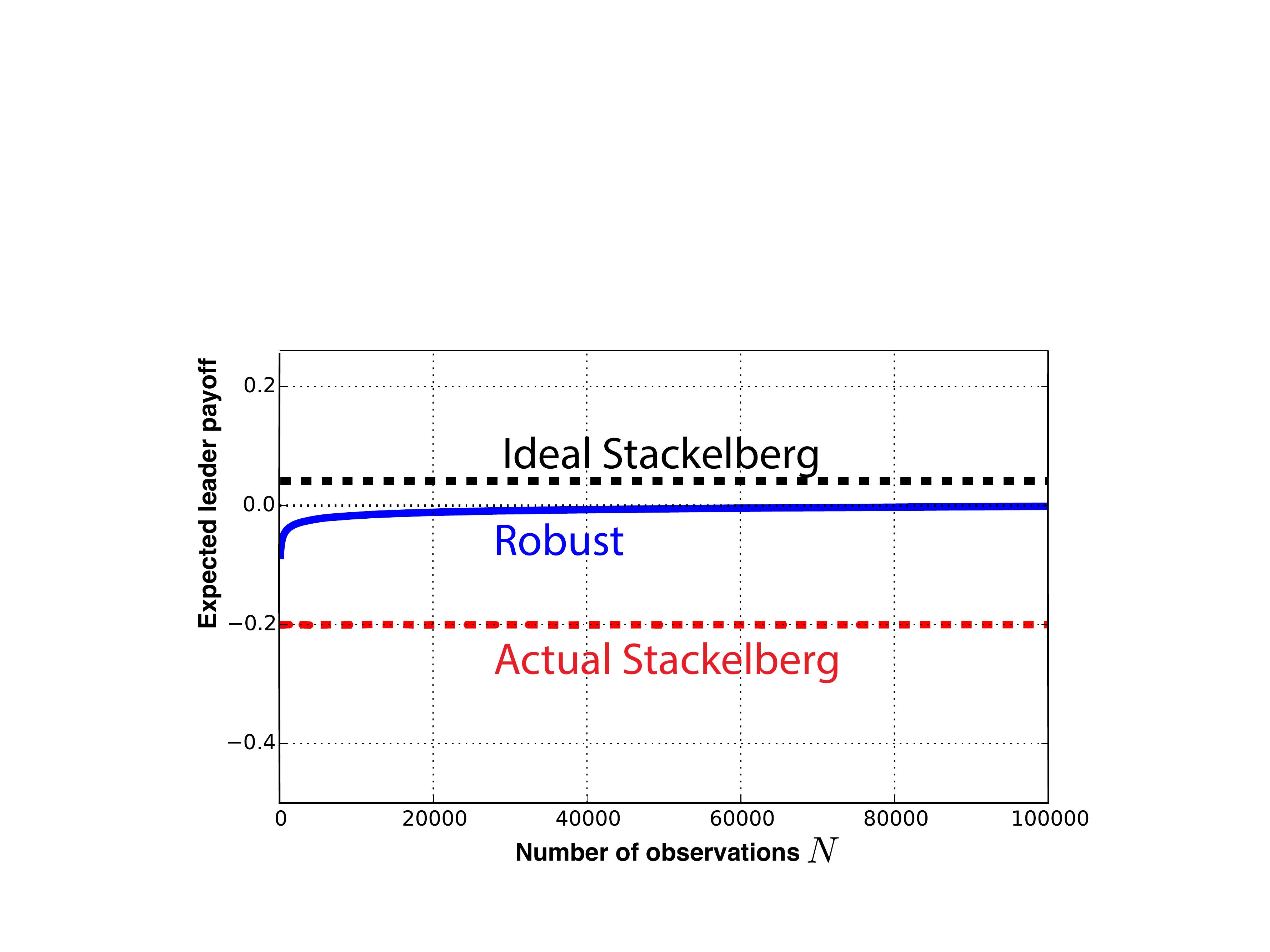}
\subfigl{0.33\textwidth}{fig:randomsecuritygame2}{Log-log plot of the gap between robust commitment payoff and idealized Stackelberg payoff.}{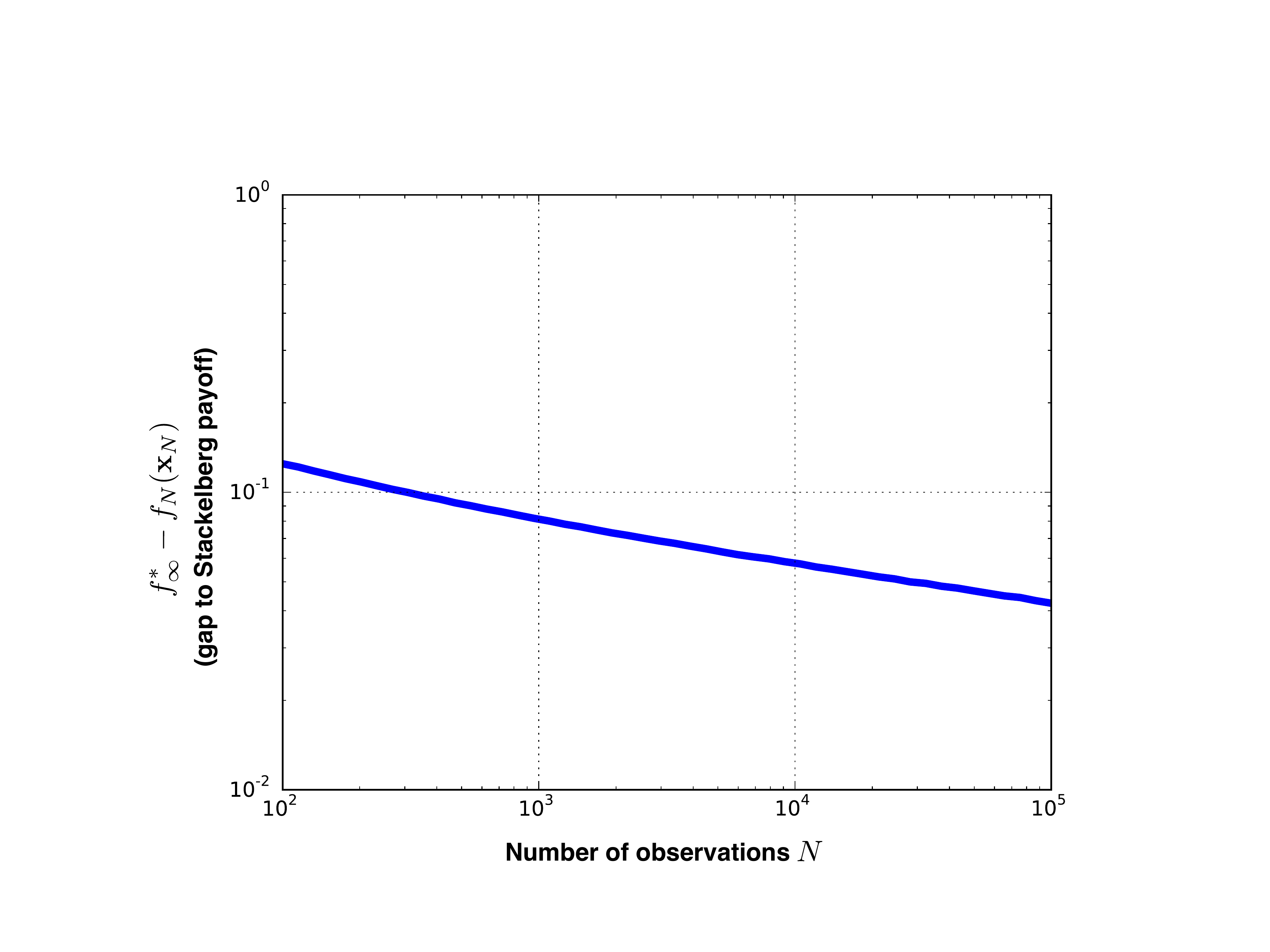}
\subfigl{0.33\textwidth}{fig:randomsecuritygame3}{Percentage plot of the gap between robust commitment payoff and idealized Stackelberg payoff.}{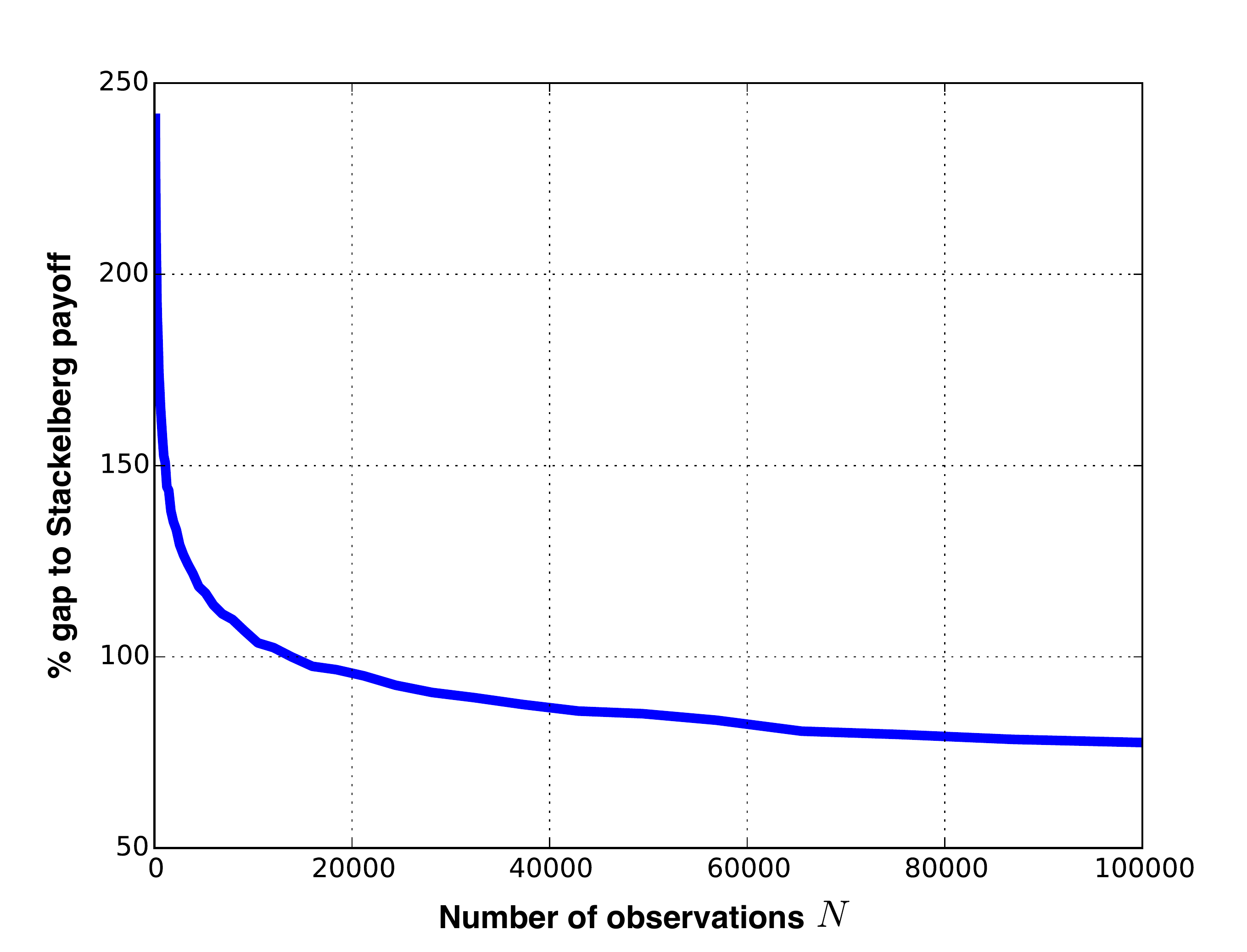}}

Our next set of simulations is inspired by the security games framework.
We create a random ensemble of $5 \times 5$ security games in which the defender can defend one of $5$ targets, and the attacker can attack one of these $5$ targets.
The defender and attacker rewards are chosen to be uniformly at random between $[0,1]$, and their penalties are uniform at random between $[-1,0]$.
This is essentially the random ensemble that was created in previous empirical work on security games~\cite{an2012security}.
Figure~\ref{fig:randomsecuritygame0} shows the construction of this ensemble.

The purpose of random security games is to show that the properties we observed above -- unstable Stackelberg commitment, robust commitment payoff approximating the optimum -- are the norm rather than the exception.
Figure~\ref{fig:randomsecuritygame} illustrates the results for random security games.
The performance of the sequence of robust commitments $\{\xvec_N\}_{N \geq 1}$, as well as the Stackelberg commitment $\xstar_{\infty}$ is plotted in Figure~\ref{fig:randomsecuritygame1} against the benchmark of idealized Stackelberg performance $\fstar_{\infty}$.
Figure~\ref{fig:randomsecuritygame2} depicts the rate of convergence of the gap in robust commitment performance to the idealized Stackelberg payoff -- we can clearly see the $\Oh(\frac{1}{\sqrt{N}})$ rate of convergence in this plot.
Finally, Figure~\ref{fig:randomsecuritygame3} plots the \textit{percentage gap} between robust commitment payoff and idealized Stackelberg payoff as a function of $N$.

We can make the following conclusions from these plots:
\begin{enumerate}
\item The Stackelberg commitment is extremely non-robust \textit{on average}.
In fact we noticed that this was the case with high probability.
This happens because the Stackelberg commitment, although it can vary widely for different games in the random ensemble, is very likely on a boundary shared with other responses and therefore unstable.
\item The robust commitments are doing much better \textit{on average} than the original Stackelberg commitment even for very large values of $N$.
The stark difference in payoff between the two motivates the construction of the robust commitment, which was as easy to compute as the Stackelberg commitment.
\end{enumerate}


\section{Proof sketches}

\quickfigure{fig:prop1illustration}{Illustration of partition of the set of follower responses, $[n]$, into sets $\{\kstar\}$ (red region), \textit{alternate best responses} (purple regions) and everything else (orange regions).}{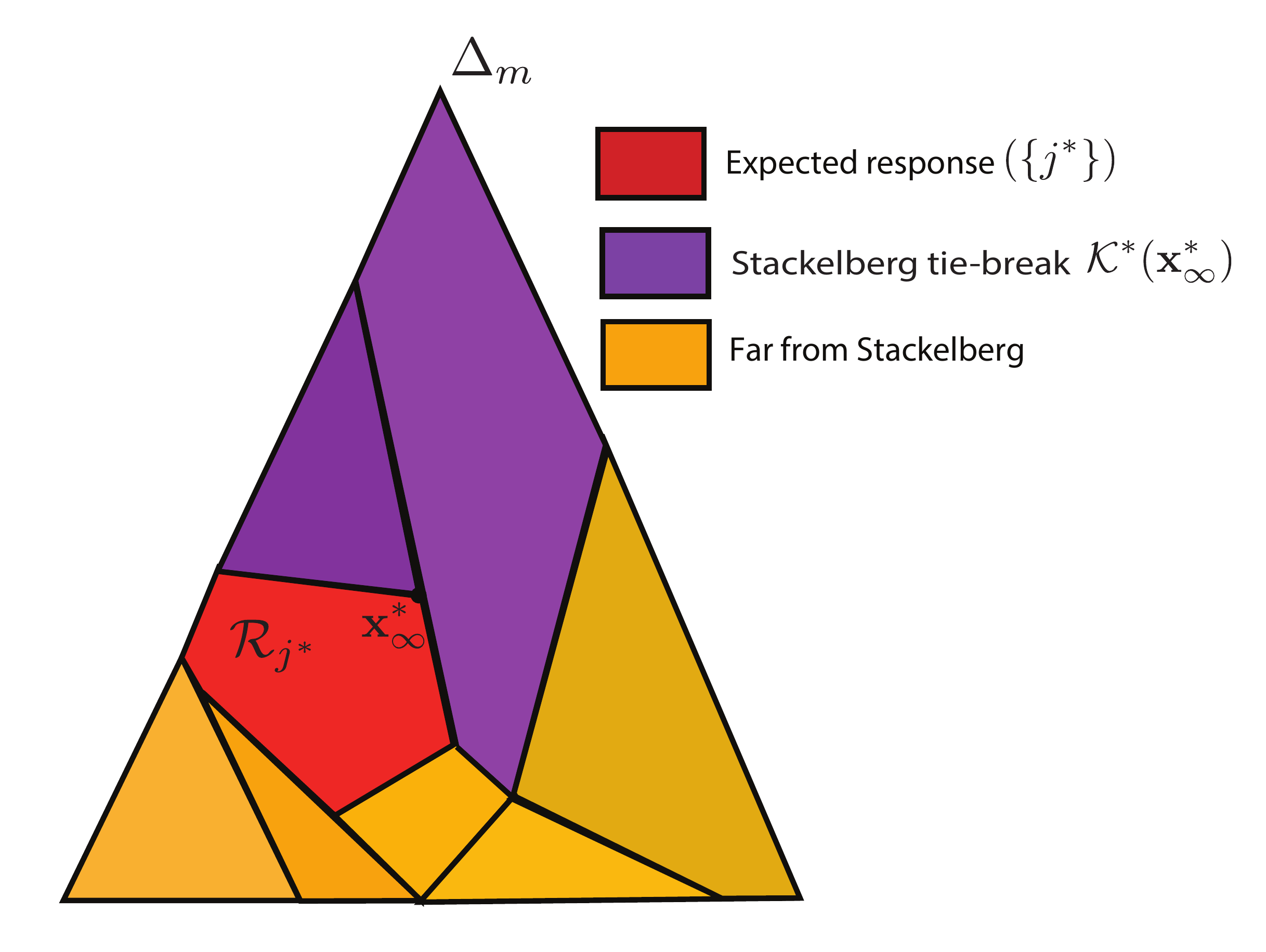}{0.5\textwidth}

In this section we describe briefly the philosophy for the proofs of our main theorems.
To understand the strong lack of robustness in Stackelberg equilibrium, it is essential to visualize the best-response-regions of the leader, i.e. subsets of the mixed strategy space for which the follower best response is a particular pure strategy.
(Note that there are $n$ such best-response regions.)
Figure~\ref{fig:prop1illustration} depicts an illustration of these best-response-regions, with the region corresponding to the follower's best response to the Stackelberg commitment highlighted in red.
The figure shows the Stackelberg commitment at a vertex (extreme-point) of the best-response polytope $\Rspace_{\kstar}$; this is generally the case~\cite{conitzer2006computing}.

First, the reason for the strong instability of Stackelberg commitment to even an infinitesmal amount of uncertainty can be seen from Figure~\ref{fig:prop1illustration}: an infinitesmal amount of fluctuation in how the leader commitment is observed will make the follower respond with a different pure strategy \textit{with constant probability}, corresponding to the regions depicted in purple.
Because of the tie-breaking assumption, it turns out that the expected payoff from any of these alternate responses is strictly worse than the Stackelberg payoff.
These facts are proved formally using the Berry-Esseen theorem.
Note that an uncertainty in commitment could also lead to a response from one of the yellow regions in the figure (which could either hurt or benefit the leader), but the probability of this happening turns out to decay exponentially.

This observation implied that the optimality of the Stackelberg commitment \textit{under ideal assumptions} was exactly what made it suboptimal under a small amount of uncertainty; we exploit this to construct the robust commitment constructions of Theorem~\ref{thm:mtimesnobs}.
The qualitative idea is to push the commitment to a small extent into the interior of the best-response-region $\Rspace_{\kstar}$ so that it simultaneously satisfies a property of being ``close" to the Stackelberg commitment, while also ensuring that the fluctuations in its empirical estimate are highly likely to stay in $\Rspace_{\kstar}$ (which guarantees that the identity of the best response of the follower is preserved).
For the special case of $m = 2$, this is a simple tradeoff to navigate as there is only one direction in which one can move into the interior.
For higher dimensions, we take inspiration from the rich literature on interior-point methods and, in fact, use Dikin ellipsoids~\cite{kannan2012random} for both the commitment construction and analysis.
Ensuring that the fluctuations of the commitment preserve the follower best response with high probability, in particular, requires sophisticated tail bounds on discrete distribution learning and a careful consideration of the best-response-polytope geometry.

The proof of Theorem~\ref{thm:approximation} ties several facts that we have seen formally, as well as alluded to, together.
First, a generalization of Theorem~\ref{thm:stackelbergnotrobust} tells us that we cannot improve sizeably over Stackelberg by committing to any mixed strategy \textit{on the boundary} between two or more best-response regions.
Second, we show that the improvement gained by a \textit{fixed commitment} in the interior of any best-response-region decreases exponentially with $N$, simply because the probability of eliciting a better-than-expected response decreases exponentially with $N$.
Putting these two facts together, the natural thing to try would be commitments that approach a boundary as $N$ increases (much like our robust commitment constructions, but now with a different motive). This should happen fast enough that we maintain a sizeable probability of eliciting a different response for every $N$, while simultaneously ensuring that that different response is actually better-than-expected.
We then show that the ensuing gain over and above Stackelberg would have to decrease with $N$ according to the rate specified.


\section{Conclusions and Discussion}
We constructed robust commitment constructions with several advantages.
First, we are able to effectively preserve the Stackelberg payoff by ensuring a high-probability guarantee on the follower responding as expected.
An oblique, but significant philosophical advantage to our robust commitments is that their guarantees hold even if we removed the pivotal assumption of follower breaking ties in favor of the leader.
We essentially showed that as the number of observations $N$ grows, our construction naturally converges to the Stackelberg commitment at a specific rate.
We also motivated that the constructions, which were inspired by interior point geometry, are computable in polynomial time given the Stackelberg commitment.

Second, we established fundamental limits on the ability of the leader to gain over and above Stackelberg payoff.
We formally showed that this ability disappears in the large-sample regime, and in a certain sense that our robust commitments are approximately optimal.
Our results established a formal connection between leader payoff and follower discrete distribution learning, and in the context of these limits, both players are mutually incentivized to increase learnability under limited samples, even though the setting is non-cooperative -- which was a rather surprising conclusion.

Our work provides implications for both leader and follower payoffs when the leader is known to be committed to a fixed strategy, but the commitment can only be revealed partially.
However, our model took commitment establishment for granted, i.e. the follower assumed that the leader would indeed be drawing its pure strategies iid from the same mixture in every round.
The partial reputation setting should most generally be modeled as a repeated game (either with a finite-horizon or discounted model), in which the belief in commitment needs to be built up over time.
Studying the problem of finite \textit{observability} of commitment in isolation is, in our view, an important first step towards eventually solving this problem, which poses many modeling challenges in itself.
In earlier rounds, directly responding to the empirical estimate of leader commitment will be suboptimal for the follower.
Instead, he may want to maintain a possibility that the leader will play minimax/Nash and respond accordingly.
From the point of view of the leader, if the iid assumption is removed, an interesting question is whether the leader could choose to play more deterministically, in such a way to increase strategy learnability while maintaining a follower impression of iid commitment.
By doing this, the leader could establish commitment faster but also run the risk of looking too deterministic/predictable in time, in which case the follower would take undue advantage.
Conversely, the leader may not even be incentivized to increase follower learnability in the finitely repeated, or discounted setting.

Finally, it is interesting to think about the applicability of the robust commitment perspective to algorithmically more difficult problems like Bayesian persuasion and public/private signalling games with multiple followers, which can have observational limitations on information transfer in much the same way as has been described for the applications in this paper.

\subsubsection*{Acknowledgments}
We thank the anonymous reviewers for valuable feedback.
We gratefully acknowledge the support of the NSF through grant AST-1444078, and the Berkeley ML4Wireless research center.

\bibliographystyle{alpha}
\bibliography{reputationreferences}

\appendix

\newpage
\section{Proofs}

Before moving into the proofs themselves, we define some additional notation.
\begin{definition}
The set of \textbf{alternate follower best response} to the mixed commitment $\xvec$ is denoted by
\begin{align*}
\Kspace^*_{\mathsf{alt}}(\xvec) := \Kspace^*(\xvec) - \{\kstar\} .
\end{align*}
\end{definition}

We will be particularly interested in this set for the Stackelberg commitment, that is, $\Kspace^*_{\mathsf{alt}}(\xstar_{\infty})$.
In general, the set will be non-empty as the follower could be agnostic between more than one pure strategy in response -- it is only responding with the pure strategy $\kstar$ to break ties in the leader's favor.
Figure~\ref{fig:prop1illustration} shows this demarcation of follower responses into the expected response $\kstar$, and alternate responses to the Stackelberg commitment $\xstar_{\infty}$.

Further, we denote maximum and minimum obtainable leader payoffs respectively by
\begin{align*}
f_{max} &:= \max_{i \in [m], j \in [n]} A_{ij} \\
f_{min} &:= \min_{i \in [m], j \in [n]} A_{ij} .
\end{align*}

\subsection{Proof of Theorem~\ref{thm:stackelbergnotrobust}}\label{sec:prop1proof}

We consider a general $2 \times n$ game and denote the Stackelberg probability of leader playing pure strategy $1$ by $\pstar_{\infty}$.
Recall that $\pstar_{\infty} \in (0,1)$ (since we have assumed for the proof that the Stackelberg commitment is mixed).
Let $j_{\mathsf{alt}}$ be \textit{the} alternate response to the Stackelberg commitment, i.e. we have $\Kspace^*(\pstar_{\infty}) = \{\kstar_{\infty},j_{\mathsf{alt}}\}$.
Without loss of generality, the best-response regions can be described as
\begin{align*}
\Rspace_{\kstar_{\infty}} &= [p^-,\pstar_{\infty}] \\
\Rspace_{j_{\mathsf{alt}}} &= (\pstar_{\infty},p^+] .
\end{align*}

Finally, we define $f^{(2)} := \lim_{\epsilon \to 0} f_{\infty}(\pstar_{\infty} + \epsilon)$.
Since we are considering leader-follower games for which the function $f_{\infty}(.)$ is discontinuous at $\pstar_{\infty}$, by the tie-breaking assumption on Stackelberg commitment we will have $f^{(2)} < \fstar_{\infty}$.

Now, we consider the quantity $f_N(\pstar_{\infty})$.
Denoting $\Phat_N$ as the empirical estimate of the quantity $\pstar$, we have
\begin{align*}
f_N(\pstar_{\infty}) &\leq \Pr\left[\Phat_N \in \Rspace_{\kstar_{\infty}} \right] \fstar_{\infty} + \Pr\left[\Phat_N \in \Rspace_{j_{\mathsf{alt}}} \right] f^{(2)} + \left(1 - \Pr\left[\Phat_N \in \Rspace_{\kstar_{\infty}} \right] - \Pr\left[\Phat_N \in \Rspace_{j_{\mathsf{alt}}} \right]\right) f_{max} \\
&= \Pr\left[\Phat_N \in (p^-,\pstar_{\infty}] \right] \fstar_{\infty} + \Pr\left[ \Phat_N \in (\pstar_{\infty},p^+] \right] f^{(2)} + \Pr\left[\Phat_N \in [0,p^-] \cup (p^+,1] \right] f_{max} \\
&= \fstar_{\infty} - \underbrace{\Pr\left[\Phat_N \in (\pstar_{\infty},p^+]\right]}_{T_1(N)} (\fstar_{\infty} - f^{(2)}) + \underbrace{\Pr\left[\Phat_N \in [0,p^-] \cup (p^+,1] \right]}_{T_2(N)} (f_{max} - \fstar_{\infty}) .
\end{align*}

We will now proceed to bound the probabilities $T_1(N)$ and $T_2(N)$.

First, we deal with the quantity $T_2(N)$, which reflects the probability of a mismatched response that is neither Stackelberg nor the alternate response on the boundary.
By the Hoeffding bound, we have
\begin{align*}
T_2(N) &:= \Pr\left[\Phat_N \in [0,p^-] \cup (p^+,1] \right] \\
&= \Pr\left[\Phat_N \in [0,p^-]\right] + \Pr\left[\Phat_N \in (p^+,1]\right] \\
&\leq \exp\{-2N(\pstar_{\infty} - p^-)^2\} + \exp\{-2N(p^+ - \pstar_{\infty})^2\}.
\end{align*}

Denoting $C'' := 2 \left(\min\{p^+ - \pstar_{\infty},\pstar_{\infty} - p^-\}\right)^2$, we then have
\begin{align}\label{eq:T2N}
T_2(N) \leq 2\exp\{-NC''\} 
\end{align}
and as expected, this probability decays exponentially with $N$.

Next, we deal with the quantity $T_1(N)$, which reflects the probability of eliciting the alternate response on the Stackelberg boundary.
We show that this event is non-vanishingly probable.

We define the following quantities
\begin{align}
S_N &:= N\Phat_N \\
Z_N &:= \frac{S_N - N\pstar_{\infty}}{\sqrt{N\pstar_{\infty}(1-\pstar_{\infty})}} .
\end{align}

Recall that $Z_N$ is a real-valued random variable.
We denote its cumulative distribution function by $F_N(.)$.

By a simple change of variables, we then have
\begin{align*}
T_1(N) &= \Pr\left[ \Phat_N \in (\pstar_{\infty},p^+]\right] \\
&= \Pr\left[Z_N \in \left(0, \frac{\sqrt{N}(p^+ - \pstar_{\infty})}{\sqrt{\pstar_{\infty}(1-\pstar_{\infty})}}\right]\right] \\
&= F_N(\frac{\sqrt{N}(p^+ - \pstar_{\infty})}{\sqrt{\pstar_{\infty}(1-\pstar_{\infty})}}) - F_N(0) .
\end{align*}

Now, recall that $S_N = \sum_{j=1}^N I_j$ for iid random variables $I_j \sim \BER(\pstar_{\infty})$.
Also note that since we have considered games with mixed Stackelberg commitment, we have $0 < \pstar_{\infty} < 1$.
We now invoke the first half of the classical Berry-Esseen theorem~\cite{berry1941accuracy,esseen1942liapounoff} stated here as a lemma.

\begin{lemma}\label{lem:berryesseen}
There exists a positive constant $C$ such that if $I_1,I_2,\ldots$ are iid random variables with $\EE[I_1] = \mu < \infty$, $\var(I_1) = \sigma^2 > 0$ and $\EE[|I_1 - \mu|^3] = \rho < \infty$, we have
\begin{align*}
|F_N(x) - \phi(x)| \leq \frac{C\rho}{\sigma^3 \sqrt{N}}
\end{align*}

for all $x \in \reals$, where $\phi(.)$ denotes the CDF of the standard normal distribution $\NORMAL(0,1)$.
\end{lemma}

It is easy to verify that the distribution $I_1 \sim \BER(\pstar_{\infty})$ satisfies the above conditions.
Therefore, we can directly apply Lemma~\ref{lem:berryesseen} and get
\begin{align*}
F_N\left(\frac{\sqrt{N}(p^+ - \pstar_{\infty})}{\sqrt{\pstar_{\infty}(1-\pstar_{\infty})}}\right) &\geq \phi(C\sqrt{N}) - \frac{C'
}{\sqrt{N}} \text{ and } \\
F_N(0) &\leq \frac{1}{2} + \frac{C'
}{\sqrt{N}} 
\end{align*}

for positive constant $C > 0$, thus giving 
\begin{align}\label{eq:T1N}
T_1(N) \geq \left(\phi(C'\sqrt{N}) - \frac{1}{2}\right) - \frac{C'}{\sqrt{N}} .
\end{align}

Substituting for the expressions for $T_1(N)$ and $T_2(N)$, we now have
\begin{align*}
\fstar_{\infty} - f_N(\pstar_{\infty}) \geq \left(\left(\phi(C'\sqrt{N}) - \frac{1}{2}\right) - \frac{C'}{\sqrt{N}}\right) C 
- 2C\exp\{-NC''\} ,
\end{align*}

which corresponds exactly to Equation~\eqref{eq:stackelbergnotrobust}.
Clearly, the right hand side of this equation is decreasing in $N$ and so the first corollary -- that $f_N(\pstar_{\infty}) \leq \fstar_{\infty}$ for $N \geq N_0$ -- holds.
Precisely, we have 

\begin{align*}
\lim_{N \to \infty}\phi(\sqrt{N}C') &= 1 \\
\lim_{N \to \infty} \frac{C'}{\sqrt{N}} &= 0 \\
\lim_{N \to \infty} 2C\exp\{-NC''\} &= 0 ,
\end{align*}

and so we have 
\begin{align*}
\fstar_{\infty} - \lim_{N \to \infty} f_N(\pstar_{\infty}) \geq \frac{\fstar_{\infty} - f^{(2)}}{2} .
\end{align*}

This is the second corollary from Theorem~\ref{thm:stackelbergnotrobust} and completes the proof.
\qed 

\subsection{Proof of Theorem~\ref{thm:mtimesnobs}}\label{sec:thm1proof}

\subsubsection{Notation}

For this proof, it will be convenient to consider the $(m-1)$-dimensional representation of the probability simplex, i.e.
\begin{align*}
\Delta_{m-1} := \{\yvec \succeq \mathbf{0} \text{ and } \inprod{\yvec}{\mathbf{1}} \leq 1\} .
\end{align*}

Then, we can represent a commitment $\xvec \in \Delta_m$ by its $(m-1)$-dimensional representation $\yvec = \begin{bmatrix}
x_1 & x_2 & \ldots & x_{m-1}
\end{bmatrix}$, and the \textit{leader payoff} if the follower were to respond with pure strategy $j \in [n]$ by 
\begin{align*}
\inprod{\yvec}{\mathbf{c}_j} + d_j
\end{align*} 

where we have
\begin{align*}
\mathbf{c}_j &:= \begin{bmatrix}
a_{j,1} - a_{j,m} \\
a_{j,2} - a_{j,m} \\
\vdots \\
a_{j,m-1} - a_{j,m} 
\end{bmatrix} \\
d_j = a_{j,m} .
\end{align*}

Similarly, we can represent the corresponding \textit{follower payoff} by
\begin{align*}
\inprod{\yvec}{\mathbf{b}'_j} + d'_j
\end{align*} 

where we have
\begin{align*}
\mathbf{b}'_j &:= \begin{bmatrix}
b_{j,1} - b_{j,m} \\
b_{j,2} - b_{j,m} \\
\vdots \\
b_{j,m-1} - b_{j,m} 
\end{bmatrix} \\
d'_j = b_{j,m} .
\end{align*}

We can also represent this representation of the empirical estimate of $\yvec$ from $N$ samples by $\Yhat_N$, and this representation Stackelberg commitment by $\ystar_{\infty}$.

Now, we can consider all the functions introduced in Section~\ref{sec:obsmodel} in terms of the commitment $\xvec$ and equivalently define them in terms of the $(m-1)$-dimensional representation of the commitment, $\yvec$.

We also denote the $p$th operator norm of a matrix by $\matsnorm{.}{p}$.

\subsubsection{The commitment construction}

We consider the $(m-1)$-dimensional representation of the best-response-region corresponding to the Stackelberg commitment, $\Rspace_{\kstar}$.
There are many things to consider while constructing a robust commitment.
The first, and obvious, one would be that the follower should respond the same way as it would to Stackelberg when it observes the full mixture.
That is, we would have $\kstar(\yvec_N) = \kstar$ or alternatively stated, $\yvec_N \in \Rspace_{\kstar}$.

Intuitively, the expected payoff of a leader commitment under observational uncertainty, particularly in terms of gap to the optimal Stackelberg payoff, will depend on two factors: one, how likely the follower is to respond the same as it would if it observed the full commitment; and two, how ``far" the leader commitment mixture is from the optimal Stackelberg commitment mixture.
We qualitatively show this dependence in the following lemma.

\begin{lemma}\label{lem:highprob}
Consider a commitment $\yvec_N$ for which we can provide the following guarantee:
\begin{align*}
\Pr[\Yhat_N \notin \Rspace_{\kstar}] \leq \epsilon_N .
\end{align*}
We then have 
\begin{align*}
\fstar_{\infty} - f_N(\yvec_N) \leq 2 (1-\epsilon_N) f_{max} \vecnorm{\yvec_N - \ystar_{\infty}}{1} + \epsilon_N (\fstar_{\infty} - f_{min}) 
\end{align*}
\end{lemma}

\begin{proof}
We have
\begin{align*}
f_N(\yvec_N) &= \sum_{j=1}^n \Pr[\Yhat_N \in \Rspace_j] (\inprod{\yvec_N}{\mathbf{c}_j} + d_j)  \\
&\geq \Pr[\Yhat_N \in \Rspace_{\kstar}] (\inprod{\yvec_N}{\mathbf{c}_{\kstar}} + d_{\kstar}) + (1-\Pr[\Yhat_N \in \Rspace_{\kstar}]) f_{min} \\
&\geq (1-\epsilon_N)\left( \inprod{\yvec_N}{\mathbf{c}_{\kstar}} + d_{\kstar} - f_{min} \right) + f_{min}\\
&= (1-\epsilon_N)\left( \inprod{\yvec_N}{\mathbf{c}_{\kstar}} + d_{\kstar}\right) + \epsilon_N f_{min} .
\end{align*}

Recall that we have $\fstar_{\infty} = \inprod{\ystar_{\infty}}{\mathbf{c}_{\kstar}} + d_{\kstar}$.
Therefore, the gap from Stackelberg is bounded as 
\begin{align*}
\fstar_{\infty} - f_N(\yvec_N) &\leq (1-\epsilon_N) \inprod{\ystar_{\infty} - \yvec_N}{\mathbf{c}_{\kstar}} + \epsilon_N (\fstar_{\infty} - f_{min}) \\
&\leq (1-\epsilon_N) \vecnorm{\mathbf{c}_{\kstar}}{\infty} \vecnorm{\yvec_N - \ystar_{\infty}}{1} + \epsilon_N (\fstar_{\infty} - f_{min}) \\
&\leq 2 (1-\epsilon_N) f_{max} |\vecnorm{\yvec_N - \ystar_{\infty}}{1} + \epsilon_N (f^*_{\infty} - f_{min}) ,
\end{align*}

where the second inequality follows from Holder's inequality.
This proves the lemma.
\end{proof}

This lemma implies that we want a commitment construction $\yvec_N$ with the following two-fold guarantee\footnote{Interestingly, the fact that $\ystar_{\infty}$ is on an extreme point of $\Rspace_{\kstar}$ will imply that the two conditions are at odds with one another, and we will need to trade them off. For instance, choosing $\yvec_N = \ystar_{\infty}$ would satisfy the second condition perfectly by being as close as possible to the Stackelberg commitment, but there would be no guarantee on the best-response as it lies on the boundary of the best-response region.}.
\begin{enumerate}
\item $\vecnorm{\yvec_N - \ystar_{\infty}}{1}$ is bounded (and ideally vanishes with $N$).
\item $\Yhat_N \in \Rspace_{\kstar}$ with high probability.
\end{enumerate}

\subsubsection{Commitment construction using localized geometry}

We will leverage the special structure of the Dikin ellipsoid~\cite{kannan2012random} used in interior-point methods to make our commitment constructions.
Observe that $\ystar_{\infty}$ is always going to be on an extreme point (vertex) of the best-response-polytope\footnote{Recall that the Stackelberg equilibrium is the solution to the LP defined on the best-response-polytope~\cite{conitzer2006computing}.} $\Rspace_{\kstar}$.
We now collect the $k = |\Kspace^*(\xstar_{\infty})|$ constraints that are satisfied \textit{with equality} at $\xstar_{\infty}$:

\begin{align*}
\inprod{\yvec}{\boldb'_j}  + d'_j &\leq \inprod{\yvec}{\boldb'_{\kstar}} + d'_{\kstar} \text{ for all } j \in \Kspace^*(\xstar_{\infty})  .
\end{align*}

This is simply the constraint set for commitments such that the follower prefers to respond with pure strategy $\kstar$ over any pure strategy $j \in \Kspace^*(\ystar_{\infty})$ (i.e. any pure strategy whose corresponding best-response-polytope shares a boundary with the Stackelberg best-response-polytope at point $\ystar$), and can be thought of as the set of \textit{local constraints to the Stackelberg vertex} in the best-response polytope $\Rspace_{\kstar}$.
We also collect the other constraints that describe $\Rspace_{\kstar}$:

\begin{align*}
\inprod{\yvec}{\boldb'_j}  + d'_j &\leq \inprod{\yvec}{\boldb'_{\kstar}} + d'_{\kstar} \text{ for all } j \notin \Kspace^*(\xstar_{\infty}) \cup \{\kstar\}  \\
\yvec &\succeq 0\\
\inprod{\mathbf{1}}{\yvec} &\leq 1 ,
\end{align*}

and together with the local constraints at the Stackelberg vertex, these describe the global constraints for the polytope.

We represent the system of inequalities in matrix form as: $B\yvec \preceq \mathbf{c}$ for some $B \in \reals^{k \times (m-1)}$ and some $\mathbf{c} \in \reals^k$.
We leverage the following useful fact about a general set of linear constraints.
\begin{fact}\label{fact:affine}
For any parameterization of linear constraints $(B, \mathbf{c})$, there exists an \textit{affine} transformation $\yvec' =  T_1 \yvec + T_2$ (where $T_1 \in \reals^{(m-1) \times (m-1)}$ is invertible and $T_2 \in \reals^{m-1}$) and a matrix $B' \in \reals^{k \times (m-1)}$ such that
\begin{align*}
B \yvec \preceq \mathbf{c} \iff B' \yvec' \preceq \mathbf{1} .
\end{align*}
We denote the transformation function by $T(\cdot)$ and its inverse by $T^{-1}(\cdot)$.
In particular, we note the relationship $B = B'T_1$.
\end{fact}

The above fact is useful\footnote{A subtle point is that there do exist special cases of polytope constraints for which Fact~\ref{fact:affine} is true only with an augmentation of the variable space from $m$ to $2m$ dimensions.
Then, defining the invertible map becomes trickier.
Nevertheless, for ease of exposition and clarity in the proof, we assume that we can indeed carry out the affine transformation without augmenting the dimension.
} because it is most convenient to define our class of commitments in the transformed space $\yvec' = T(\yvec)$.
\begin{definition}\label{def:robustcommitment}
For a particular value of $\delta \in (0,1)$, Stackelberg commitment $\ystar_{\infty}$, and local constraints modeled by $(B, \mathbf{c})$, we define a $\delta$-deviation commitment by
\begin{align*}
\yvec(\delta;\ystar_{\infty}) &:= T^{-1}(\yvec'(\delta;(\ystar)'_{\infty})) \text{ where } \\
\yvec'(\delta;(\ystar)'_{\infty}) &:= (1 - \delta) (\ystar)'_{\infty} .
\end{align*}
\end{definition}


Our robust commitments $\{\yvec_N\}_{N \geq 1}$ are going to be taken out of the set of $\delta$-deviation commitments, with appropriately chosen values of $\{\delta_N\}_{N \geq 1}$.
\textit{Clearly, the computational complexity of constructing any $\delta$-deviation commitment is equivalent to the complexity of computing the Stackelberg equilibrium itself.}

To understand how to set these values, we will turn to the question of how to satisfy the three conditions above.

First, we observe that $\yvec(\delta;\ystar_{\infty})$ satisfies the \textit{local constraints} $B\yvec \preceq \mathbf{c}$ for any $\delta \in (0,1)$.
Because of Fact~\ref{fact:affine}, it suffices to show that its affine transformation $\yvec'(\delta;(\ystar)'_{\infty})$ satisfies the local constraints $B' \yvec' \preceq \mathbf{1}$.
Recall that $(\ystar)'_{\infty}$ satisfies \textit{all the local constraints} with equality, i.e. we have $B' (\ystar)'_{\infty} = \mathbf{1}$.
From the definition of the commitment, we thus have
\begin{align*}
B' \yvec'(\delta;(\ystar)'_{\infty}) &= (1 - \delta) B' (\ystar)'_{\infty} \\
&= (1 - \delta) \mathbf{1} \preceq \mathbf{1} .
\end{align*}

Next, we turn to the question of how close such a defined commitment would be from the Stackelberg commitment $\ystar_{\infty}$, in terms of the $\ell_1$ norm.
For this, we have
\begin{align*}
\vecnorm{\yvec(\delta;\ystar_{\infty}) - \ystar_{\infty}}{1} &= \vecnorm{T_1^{-1}(\yvec'(\delta;(\ystar)'_{\infty}) - (\ystar)'_{\infty})}{1} \\
&= \delta \vecnorm{T^{-1} (\ystar)'_{\infty}}{1} \\
&= \delta \vecnorm{\ystar_{\infty}}{1} \leq \delta .
\end{align*}

Therefore, we have
\begin{align}\label{eq:robustcommitmentclose}
\vecnorm{\yvec(\delta;\ystar_{\infty}) - \ystar_{\infty}}{1} \leq \delta .
\end{align}

In lieu of Lemma~\ref{lem:highprob}, we wish to choose values $\{\delta_N\}_{N \geq 1}$ (to create commitments $\{\yvec_N\}_{N \geq 1}$) such that $\delta_N$ decreases with $N$ sufficiently fast, while maintaining a high probability of staying in the best-response polytope $\Rspace_{\kstar}$.
To understand the rate at which we can decrease $\delta_N$, we need to prove a high-probability best-response guarantee.

\subsubsection{Using the \textbf{local} Dikin ellipsoid as a confidence ball}

For a (affine-transformed) commitment $\yvec'(\delta;(\ystar)'_{\infty})$, we make use of the \textit{local Dikin ellipsoid} centered at $\yvec'(\delta;(\ystar)'_{\infty})$, defined below for an arbitrary point $\yvec'$.
\begin{definition}[~\cite{kannan2012random}]
For constraint set $B' \yvec' \preceq \mathbf{1}$, the \textbf{Dikin ellipsoid} of radius $r$ centered at $\yvec'$ is given by
\begin{align}\label{eq:dikindef}
\mathbb{B}_{B', \mathbf{1}, \yvec'}(r) := \{\mathbf{z}': (\mathbf{z}' - \yvec')^\top H(\yvec') (\mathbf{z}' - \yvec') \leq r\} ,
\end{align}

where we define
\begin{align}\label{eq:hdefinition}
H(\yvec') := \sum_{i=1}^k \frac{(\boldb')_i (\boldb')_i^\top}{(1 - \inprod{(\boldb')_i}{\yvec'})^2} .
\end{align}

The Dikin ellipsoid has two special properties~\cite{kannan2012random}:
\begin{enumerate}
\item \textbf{Affine invariance}: (using the notation from Fact~\ref{fact:affine}) For transformation $\yvec' = T(\yvec)$, the Dikin ellipsoid of radius $r$ centered at the point $\yvec$ for the polytope $B \yvec \preceq \mathbf{c}$ is $\mathbb{B}_{B, \mathbf{c}, \yvec'}(r) = T^{-1}(\mathbb{B}_{B', \mathbf{1},  \yvec'}(r))$.
\item \textbf{Interior guarantee}: For any \textit{interior} point $\yvec'$ (according to the constraint set $B' \yvec' \preceq \mathbf{1}$), the Dikin ellipsoid of radius $1$ centered at $\yvec'$ is contained in the feasibility set, that is,
\begin{align*}
\mathbf{z}' \in \mathbb{B}_{B', \mathbf{1}, \yvec'}(1) \implies B' \mathbf{z}' \preceq \mathbf{1} .
\end{align*}
\end{enumerate}
\end{definition}

We center our Dikin ellipsoid at $\yvec'(\delta;(\ystar)'_{\infty})$, and observe that the constraint takes on a particularly nice form, as stated by the following simple lemma.

\begin{lemma}\label{lem:dikinweightedball}
For any $\delta \in (0,1)$, the Dikin ellipsoid can be expressed as
\begin{align}\label{eq:dikinaffinetransformed}
\mathbb{B}_{B', \mathbf{1},\yvec'(\delta;(\ystar)'_{\infty})}(1) = \{\mathbf{z}': \vecnorm{B'(\mathbf{z}' - \yvec'(\delta;(\ystar)'_{\infty}))}{2} \leq \delta \} .
\end{align}

Furthermore, in the original space we can write
\begin{align}\label{eq:dikinoriginal}
\mathbb{B}_{B, \mathbf{c},\yvec(\delta;\ystar_{\infty})}(1) = \{\mathbf{z}: \vecnorm{B(\mathbf{z} - \yvec(\delta;\ystar_{\infty}))}{2} \leq \delta \} .
\end{align}
\end{lemma}

\begin{proof}
From Definition~\ref{def:robustcommitment}, we observe that $B' \yvec'(\delta;(\ystar)'_{\infty}) = (1 - \delta) B' (\ystar)'_{\infty} = (1 - \delta) \mathbf{1}$.
This implies that
\begin{align*}
1 - \inprod{(\boldb')_i}{\yvec'(\delta;(\ystar)'_{\infty})} = 1 - (1 - \delta) = \delta ,
\end{align*}

and thus we have
\begin{align*}
H(\yvec'(\delta;(\ystar)'_{\infty})) &= \frac{\sum_{i=1}^k (\boldb')_i (\boldb')_i^\top}{\delta^2} \\
&= \frac{(B')^\top B'}{\delta^2} 
\end{align*}

where in the last equality step, we have used $(B')^\top B' = \sum_{i=1}^k (\boldb')_i (\boldb')_i^\top$, noting that $(\boldb')_i$ denotes the $i^{th}$ row of $B'$.

Thus, the ellipsoid constraint in Equation~\eqref{eq:dikindef} can be rewritten as 
\begin{align*}
\frac{1}{\delta^2} (\mathbf{z}' - \yvec'(\delta;(\ystar)'_{\infty}))^\top (B')^\top B' (\mathbf{z}' - \yvec'(\delta;(\ystar)'_{\infty})) &\leq 1 \\
\implies  \vecnorm{B'(\mathbf{z}' - \yvec'(\delta;(\ystar)'_{\infty}))}{2}^2 &\leq \delta^2 \\
\implies \vecnorm{B'(\mathbf{z}' - \yvec'(\delta;(\ystar)'_{\infty}))}{2} &\leq \delta ,
\end{align*}

thus completing the first part of the proof (Equation~\eqref{eq:dikinaffinetransformed}).

For the second part of the proof, we use the affine invariance property of the Dikin ellipsoid, which tells us that 
\begin{align*}
\zvec \in \mathbb{B}_{B, \mathbf{c}, \yvec}(1) \implies \zvec' = T_1 \zvec + T_2 \in \mathbb{B}_{B', \mathbf{1}, \yvec'}(1) \\
\implies \vecnorm{B'(\mathbf{z}' - \yvec'(\delta;(\ystar)'_{\infty}))}{2} &\leq \delta .
\end{align*}

Now, observe that 
\begin{align*}
B'(\zvec' -  \yvec'(\delta;(\ystar)'_{\infty})) &= B'(T_1 \zvec + T_2 - T_1  \yvec(\delta;\ystar_{\infty}) - T_2) \\
&= (B' T_1) (\zvec -  \yvec(\delta;\ystar_{\infty})) \\
&= B (\zvec -  \yvec(\delta;\ystar_{\infty}))
\end{align*}

where in the last step we have used the relationship $B = B'T_1$ from Fact~\ref{fact:affine}.
Putting these observations together, we have
\begin{align*}
\zvec \in \mathbb{B}_{B, \mathbf{c}, \yvec(\delta;\ystar_{\infty})}(1) \implies \vecnorm{B (\zvec -  \yvec(\delta;\ystar_{\infty}))}{2} \leq \delta ,
\end{align*}

completing the second part of the proof.
\end{proof}

At this stage, it is worth remembering that the commitment is \textit{mixed}, and the payoff from using a $\delta$-deviation commitment $\yvec(\delta;\ystar_{\infty}) \in \Delta_{m-1}$ under a finite number of observations $N$ depends on the guarantee that its observed empirical distribution $\Yhat_N$ (typically) stays inside the best-response region.
As a starting point we need to guarantee that at least the \textit{local vertex constraints} are not violated.

Note that $\yvec(\delta;\ystar_{\infty}) \in \Delta_{m-1}$ is an interior point for any $\delta > 0$, and thus the interior guarantee property of the Dikin ellipsoid can be applied.
We thus know that if the empirical distribution of the commitment stays inside the Dikin ellipsoid centered at the actual commitment, it will stay inside the local constraint feasibility set.
Thus, it makes sense to use the Dikin ellipsoid as a confidence ball and tail bound the probability that the empirical estimate lies outside this ball.
Because of the weighted $\ell_2$-ball structure on the particular ellipsoid corresponding to a $\delta$-deviation commitment that we proved in Lemma~\ref{lem:dikinweightedball}, this is not difficult to do.
We state this formally in the following lemma.

\begin{lemma}\label{lem:weightedballconc}
For a given $\delta > 0$, let $\Yhat_N$ be the empirical distribution of $N$ samples drawn from the $\delta$-deviation commitment $\yvec(\delta;\ystar_{\infty})$.
Then, we have
\begin{align*}
\Pr \left[\Yhat_N \notin  \mathbb{B}_{B, \mathbf{c}, \yvec(\delta;\ystar_{\infty})}(1) \right] \leq 3 \exp\{-\frac{N \delta^2}{ 25 \matsnorm{B}{\mathsf{op}}^2}\}
\end{align*}

provided that $N \geq \frac{20 m\matsnorm{B}{\mathsf{op}}^2}{\delta^2}$.
\end{lemma}

\begin{proof}
The proof is a simple consequence of Devroye's lemma~\cite{devroye1983equivalence}, which tail bounds the total variation between the empirical estimate of a discrete distribution and the true distribution.
\begin{lemma}[\cite{devroye1983equivalence}]\label{lem:devroye}
Let $\Yhat_N$ be the empirical distribution of $N$ samples drawn from any distribution $\yvec \in \Delta_{m-1}$. Then, as long as $\delta \geq \sqrt{\frac{20m}{N}}$ we have
\begin{align*}
\Pr\left[\vecnorm{\Yhat_N - \yvec}{1} \geq \delta \right] \leq 3\exp\{\frac{N\delta^2}{25}\} .
\end{align*}
\end{lemma}

We note from Lemma~\ref{lem:dikinweightedball} that 
\begin{align*}
\Yhat_N \notin  \mathbb{B}_{B, \mathbf{c}, \yvec(\delta;\ystar_{\infty})}(1) \implies \vecnorm{B (\Yhat_N -  \yvec(\delta;\ystar_{\infty}))}{2} > \delta ,
\end{align*}

and thus, we have
\begin{align*}
\Pr \left[\Yhat_N \notin  \mathbb{B}_{B, \mathbf{c}, \yvec(\delta;\ystar_{\infty})}(1) \right] &= \Pr\left[ \vecnorm{B (\Yhat_N -  \yvec(\delta;\ystar_{\infty}))}{2} > \delta \right] \\
&\stackrel{\1}\leq \Pr\left[ \vecnorm{B}{op} \vecnorm{\Yhat_N -  \yvec(\delta;\ystar_{\infty}))}{2} > \delta \right] \\
&\stackrel{\2}\leq \Pr\left[ \vecnorm{B}{op} \vecnorm{\Yhat_N -  \yvec(\delta;\ystar_{\infty}))}{1} > \delta \right] \\
&= \Pr\left[  \vecnorm{\Yhat_N -  \yvec(\delta;\ystar_{\infty}))}{1} > \delta/\vecnorm{B}{op} \right]
\end{align*}

where inequality $\1$ uses the definition of the operator norm and inequality $\2$ uses the fact that $\vecnorm{\mathbf{v}}{2} \leq \vecnorm{\mathbf{v}}{1}$ for any finite-dimensional vector $\mathbf{v}$.
Applying Lemma~\ref{lem:devroye} directly then gives us
\begin{align*}
\Pr \left[\Yhat_N \notin  \mathbb{B}_{B, \mathbf{c}, \yvec(\delta;\ystar_{\infty})}(1) \right] \leq 3 \exp\{-\frac{N \delta^2}{ 25 \matsnorm{B}{\mathsf{op}}^2}\} 
\end{align*}

as long as 
\begin{align*}
\frac{\delta}{\matsnorm{B}{\mathsf{op}}} &\geq \sqrt{\frac{20m}{N}} \\
\implies N &\geq \frac{20 m\matsnorm{B}{\mathsf{op}}^2}{\delta^2} .
\end{align*}

This completes the proof.

\end{proof}

\subsubsection{Completing proof of Theorem~\ref{thm:mtimesnobs}: Ensuring global constraint satisfiability}

Let us recap what we have proved so far about a $\delta$-deviation commitment $\yvec(\delta;\ystar_{\infty})$ for any $\delta \in (0,1)$.
\begin{enumerate}
\item For $N$ samples from $\yvec(\delta;\ystar_{\infty})$, we have $\Pr \left[\Yhat_N \notin  \mathbb{B}_{B, \mathbf{c}, \yvec(\delta;\ystar_{\infty})}(1) \right] \leq 3 \exp\{-\frac{N \delta^2}{ 25 \matsnorm{B}{\mathsf{op}}^2}\}$ (from Lemma~\ref{lem:weightedballconc}).
\item $\vecnorm{\yvec(\delta;\ystar_{\infty}) - \ystar_{\infty}}{1} \leq \delta$.
\end{enumerate}

Thus, from Lemma~\ref{lem:highprob} we have for any $\delta$-deviation commitment,
\begin{align*}
\fstar_{\infty} - f_N(\yvec(\delta;\ystar_{\infty})) \leq 2 \delta f_{max} + \Pr\left[\Yhat_N \notin \Rspace_{\kstar} \right] (\fstar_{\infty} - f_{min}) 
\end{align*}

Thus, if we had $\mathbb{B}_{B, \mathbf{c}, \yvec(\delta;\ystar_{\infty})}(1) \subset \Rspace_{\kstar}$, we would have 
\begin{align*}
\Pr\left[\Yhat_N \notin \Rspace_{\kstar} \right] \leq \Pr \left[\Yhat_N \notin  \mathbb{B}_{B, \mathbf{c}, \yvec(\delta;\ystar_{\infty})}(1) \right] \leq 3 \exp\{-\frac{N \delta^2}{ 25 \matsnorm{B}{\mathsf{op}}^2}\} .
\end{align*}

However, the set $\Rspace_{\kstar}$ includes \textit{global constraints} in addition to the local constraints $B \yvec \preceq \mathbf{c}$, and all points in the \textit{local} Dikin ellipsoid need not satisfy these constraints.
This is the final technicality in the proof that we now deal with.
We will see that for a small enough value of $\delta$ (that depends on how the local geometry of the polytope relates to the global geometry), we can guarantee global satisfiability.
Let the constraints corresponding to the convex polytope $\Rspace_{\kstar}$ be represented by $C \yvec \preceq \mathbf{d}$, and the corresponding constraints after the affine transformation $(T_1,T_2)$ be represented as $C' \yvec' \preceq \mathbf{d}'$ (where the values of $\mathbf{d}'$ corresponding the local constraints are $1$).
Thus, for the vertex $(\ystar)'_{\infty}$, we can define the quantity
\begin{align*}
\mathcal{Z}(\Rspace_{\kstar};(\ystar)'_{\infty}) := \sup \{\delta > 0: \zvec' \in \mathbb{B}_{B', \mathbf{1}, \yvec'(\delta;(\ystar)'_{\infty})}(1) \implies C' \zvec' \preceq \mathbf{d}' \} .
\end{align*}

Because $\Rspace_{\kstar}$ is \textit{non-empty and convex}, we have $\mathcal{Z}(\Rspace_{\kstar};(\ystar)'_{\infty}) > 0$.

From this definition, under the condition $\delta < \mathcal{Z}(\Rspace_{\kstar};(\ystar)'_{\infty})$ we have
\begin{align*}
\mathbb{B}_{B', \mathbf{1}, \yvec'(\delta;(\ystar)'_{\infty})}(1) \subset T(\Rspace_{\kstar}) \\
\implies \mathbb{B}_{B, \mathbf{1}, \yvec(\delta;\ystar_{\infty}}(1) \subset \Rspace_{\kstar} ,
\end{align*}

where the last implication is because of the affine-invariance property of the Dikin ellipsoid.

On the other hand, we used the condition $N \geq \frac{20 m\matsnorm{B}{\mathsf{op}}^2}{\delta^2}$ to prove Lemma~\ref{lem:weightedballconc}.
Combining these inequalities tells us that we require $N > \frac{20 m\matsnorm{B}{\mathsf{op}}^2}{\mathcal{Z}(\Rspace_{\kstar};(\ystar)'_{\infty})^2} = \Oh(m)$ to prove our result.

Then, we formally define our robust commitment for a particular value of $N$ below, and prove this final lemma which is essentially a formal statement of Theorem~\ref{thm:mtimesnobs}.
\begin{lemma}\label{lem:mainthmformal}
For every $N > \frac{20 m\matsnorm{B}{\mathsf{op}}^2}{\mathcal{Z}(\Rspace_{\kstar};(\ystar)'_{\infty})^2}$, and every $p < 1/2$, we define the $p$-robust commitment as a $\delta_{N,p}$-deviation commitment $\yvec_{N,p} := \yvec(\delta_{N,p}; \ystar_{\infty})$, where
\begin{align}\label{eq:deltanp}
\delta_{N,p} &:= \mathcal{Z}(\Rspace_{\kstar};(\ystar)'_{\infty}) \left(\frac{m}{N}\right)^p .
\end{align}

We then have
\begin{align*}
f_N(\yvec_{N,p}) &\leq \frac{2 f_{max}} \cdot \mathcal{Z}(\Rspace_{\kstar};(\ystar)'_{\infty}) \cdot \left(\frac{m}{N}\right)^p + 3 \exp\{-\frac{m^{2p} \cdot \mathcal{Z}(\Rspace_{\kstar};(\ystar)'_{\infty})^2 \cdot N^{1-2p}}{25 \matsnorm{B}{\mathsf{op}}^2}\} (\fstar_{\infty} - f_{min})  \\
&= \Oh\left(\left(\frac{m}{N}\right)^p + \exp\{-\omega(1) \cdot N^{1- 2p}\}\right) .
\end{align*}

\end{lemma}

\begin{proof}
This is a simple consequence of everything put together.
Since $N > m$, we have $\delta_{N,p} <  \mathcal{Z}(\Rspace_{\kstar};(\ystar)'_{\infty})$ and thus we have $\mathbb{B}_{B, \mathbf{1}, \yvec(\delta_{N,p};\ystar_{\infty}}(1) \subset \Rspace_{\kstar}$.
This tells us that
\begin{align*}
\Pr\left[\Yhat_N \notin \Rspace_{\kstar} \right] \leq 3 \exp\{-\frac{N \delta_{N,p}^2}{ 25 \matsnorm{B}{\mathsf{op}}^2}\} .
\end{align*}

and thus from Lemma~\ref{lem:highprob} we get the following expression:
\begin{align*}
\fstar_{\infty} - f_N(\yvec_N) \leq 2 \delta_{N,p} f_{max} + 3 \exp\{-\frac{N \delta_{N,p}^2}{ 25 \matsnorm{B}{\mathsf{op}}^2}\}  (\fstar_{\infty} - f_{min}) .
\end{align*}

Directly substituting the expression for $\delta_{N,p}$ in Equation~\eqref{eq:deltanp} into the above expression completes the proof.
\end{proof}

\subsection{Proof of Theorem~\ref{thm:approximation}}

\quickfigure{fig:thm3illustration_proof}{Illustration of partition of the set of follower responses, $[n]$, into sets $\{\kstar\}$ (red region), $\Kspace^*_{\mathsf{aug}}$ (blue regions) and $\Kspace^*_{\mathsf{far}}$ (yellow regions).}{Thm3illustration}{0.5\textwidth}

Recall that $\fstar_N := \max_{\xvec \in \Delta_m} f_N(\xvec)$.
To prove an upper bound on $\fstar_N$, we will upper bound $f_N(\xvec)$ for every $\xvec \in \Delta_m$.

Without loss of generality the same proof method will extend to all $\xvec \in \Delta_m$.
Denoting as shorthand $p_j(\xvec) := \Pr\left[\Xhat_N \in \Rspace_j \right]$, we have
\begin{align}\label{eq:fNxvec}
f_N(\xvec) &= \sum_{j=1}^n p_j(\xvec) \inprod{\bolda_j}{\xvec} \\
&= \sum_{j=1}^n T_j(\xvec) 
\end{align}

where we denote $T_j(\xvec) := p_j(\xvec) \inprod{\bolda_j}{\xvec}$.
We will proceed to upper bound the quantity $T_j(\xvec)$ for every $\xvec \in \Delta_m$ and every $j \in [n]$.

To do this, we will see that it is natural to divide all the pure strategy responses that are possible to a commitment $\xvec$ into three categories.
The first is the expected response $\kstar(\xvec)$.
The second is the set of responses whose regions are \textit{adjacent} to the expected response as defined below.

\begin{definition}
For a particular commitment $\xvec \in \Delta_m$, the set of \textbf{adjacent-to-expected} responses $\Kspace^*_{aug}(\xvec)$ is the set of all best-responses whose corresponding best-response-regions share a boundary with the best-response-region corresponding to the best response to $\xvec$.
Formally, we have
\begin{align*}
\Kspace^*_{\mathsf{aug}}(\xvec) := \{j \in [n]: j \neq \kstar(\xvec) \text{ and } \text{cl}(\Rspace_{\kstar(\xvec)}) \cap \text{cl}(\Rspace_j) \neq \emptyset \} .
\end{align*}

We also define $\Kspace_{\mathsf{far}} := [n] - (\{\kstar(\xvec)\} \cup \Kspace^*_{\mathsf{aug}}(\xvec))$ as the set of all follower responses that are ``far" from the expected response in this sense.
\end{definition}

The illustration in Figure~\ref{fig:thm3illustration_proof} shows this division.

For the rest of the proof, we will drop the term $\xvec$ from the notation and denote $\Kspace^*_{\mathsf{aug}} := \Kspace^*_{\mathsf{aug}}(\xvec)$ as well as $\kstar := \kstar(\xvec)$.
This is done for notational simplicity.

It is first easy to show a bound on $T_{\kstar}(\xvec)$. 
In particular, we can directly use the definition of the function $f_{\infty}(.)$ to obtain
\begin{align}\label{eq:Tkstar}
T_{\kstar}(\xvec) &= p_{\kstar}(\xvec) \inprod{\bolda_{\kstar}}{\xvec} \\
&= p_{\kstar}(\xvec) f_{\infty}(\xvec) \\
&\leq p_{\kstar}(\xvec) f^*_{\infty} .
\end{align}

This inequality is also intuitive because the leader would only hope to gain from eliciting a different-than-expected response.
Next, we deal with this cases.

\subsubsection{``Far"-from-expected responses}

We collect the set of commitments that (if observed fully) would elicit a response far away from the actual expected response. 
Formally, we denote $\Rspace_{\mathsf{far}} := \cup_{j \in \Kspace_{\mathsf{far}}} \Rspace_j$.
Now, we wish to bound the term
\begin{align*}
T_{far} := \sup_{\xvec \in \Delta_m} \sum_{j \in \Kspace_{\mathsf{far}}} T_j(\xvec) .
\end{align*}
By definition, we have $\text{cl}(\Rspace_{\kstar}) \cap \text{cl}(\Rspace_{far}) = \emptyset$.
Because we are considering \textit{finite} games, i.e. $n < \infty$, there exists a constant $C > 0$ that depends solely on the parameters of the game such that 
\begin{align}\label{eq:boundedkull}
\inf_{\xvec \in \Rspace_{\kstar},\xvec' \in \Rspace_{far}} \kull{\xvec'}{\xvec} \geq C .
\end{align}

Geometrically, Figure~\ref{fig:thm3illustration} shows this separation between the expected-response-region and any far-from-expected-response-region.
To understand the probability of eliciting such responses, we invoke a classical result from large-deviations theory, Sanov's theorem~\cite{csiszar2011information}.
The upper bound part of the theorem is restated here as a lemma and with appropriate notation.
\begin{lemma}
Let $I_1,I_2,\ldots,I_N$ be $\text{i.i.d } \sim \xvec$ for any $\xvec \in \Delta_m$ and $\Xhat_N$ denote the empirical estimate.
Then, for any region $\Rspace \subseteq \Delta_m$, we have
\begin{align}\label{eq:sanov}
\Pr\left[\Xhat_N \in \Rspace \right] \leq (N+1)^m 2^{-N\inf_{\xvec' \in \Rspace} \kull{\xvec'}{\xvec}} .
\end{align}
\end{lemma}

Combining equations~\eqref{eq:sanov} and~\eqref{eq:boundedkull}, we therefore get
\begin{align}\label{eq:Tjfar}
T_{far} &\leq \left[\sup_{\xvec \in \Delta_m} \sum_{j \in \Kspace_{\mathsf{far}}} p_j(\xvec)\right] f_{max} \\
&\leq \left[(N+1)^m 2^{-N\inf_{\xvec \in \Rspace_{\kstar}, \xvec' \in \Rspace_{far}} \kull{\xvec'}{\xvec}}\right] f_{max}\\
&\leq (N+1)^m 2^{-NC} f_{max} \\
&= C(N+1)^m \exp\{-NC\} f_{max} . 
\end{align}

The rationale for calling these responses \textit{far-from-expected} is now clear: there is a minimum constant separation in terms of the KL-divergence from the expected best response, and so the probability of realizing these responses decreases exponentially with $N$.

Dealing with the adjacent-to-expected responses is more delicate.
We turn to this case next.

\subsubsection{Adjacent-to-expected responses}

Consider the set of adjacent-to-expected response $\Kspace^*_{\mathsf{aug}}$.
We wish to bound the term $\sum_{j \in \Kspace^*_{\mathsf{aug}}} T_j(\xvec)$.
It turns out that we can no longer control the probability that one of these responses is elicited for all choices of $\xvec \in \Rspace_{\kstar}$ -- this is because the commitment $\xvec$ could be chosen arbitrarily close to a boundary of its expected-response-region.
However, we can bound the ensuing payoff as a function of how close the commitment is to a boundary.
This notion of closeness is defined in terms of the $\ell_1$-norm below.
\begin{definition}
For a commitment $\xvec \in \Rspace_{\kstar}$ and a particular adjacent response $j \in \Kspace^*_{\mathsf{aug}}$, we define its minimum distance to the boundary by
\begin{align*}
\delta_1(\xvec;j) := \inf_{\xvec' \in \text{cl}(\Rspace_j)} \vecnorm{\xvec - \xvec'}{1} .
\end{align*}
\end{definition}

First, we use this notion to bound the maximum possible payoff that could be elicited.
\begin{lemma}\label{lem:boundinggain}
For any commitment $\xvec \in \Rspace_{\kstar}$, we have
\begin{align*}
T_j(\xvec) \leq p_j(\xvec)\left[\fstar_{\infty} + f_{max} \delta_1(\xvec;j) \right] .
\end{align*}
\end{lemma}

\begin{proof}
Let $\xtilde \in {\arg \min}_{\xvec' \in \text{cl}(\Rspace_j)} \vecnorm{\xvec - \xvec'}{1}$.
(Note that the minimum exists because we've taken the closure of the region.)
Using Holder's inequality, we have
\begin{align*}
\inprod{\bolda_j}{\xvec - \xtilde} &\leq \vecnorm{\bolda_j}{\infty} \vecnorm{\xvec - \xtilde}{1} \\
&\leq f_{max}\delta_1(\xvec;j) .
\end{align*}

For every $j \in \Kspace^*_{aug}$ we have
\begin{align*}
\inprod{\bolda_j}{\xvec} &\leq \inprod{\bolda_j}{\xtilde} + f_{max}\delta_1(\xvec;j) \\
&\leq f_{\infty}(\xtilde) + f_{max}\delta_1(\xvec;j) \\
&\leq \fstar_{\infty} + f_{max}\delta_1(\xvec;j) .
\end{align*}

where we are crucially using the fact that $\xtilde$ lies on the boundary and the tie-breaking assumption, to tie its payoff to the function $f_{\infty}(.)$.
Substituting the above bound into the definition of $T_j(\xvec)$ completes the proof.
\end{proof}

\quickfigure{fig:thm3illustration2}{Illustration showing the potential gain in payoff obtainable by eliciting a different-than-expected response for a $2 \times 2$ game.}{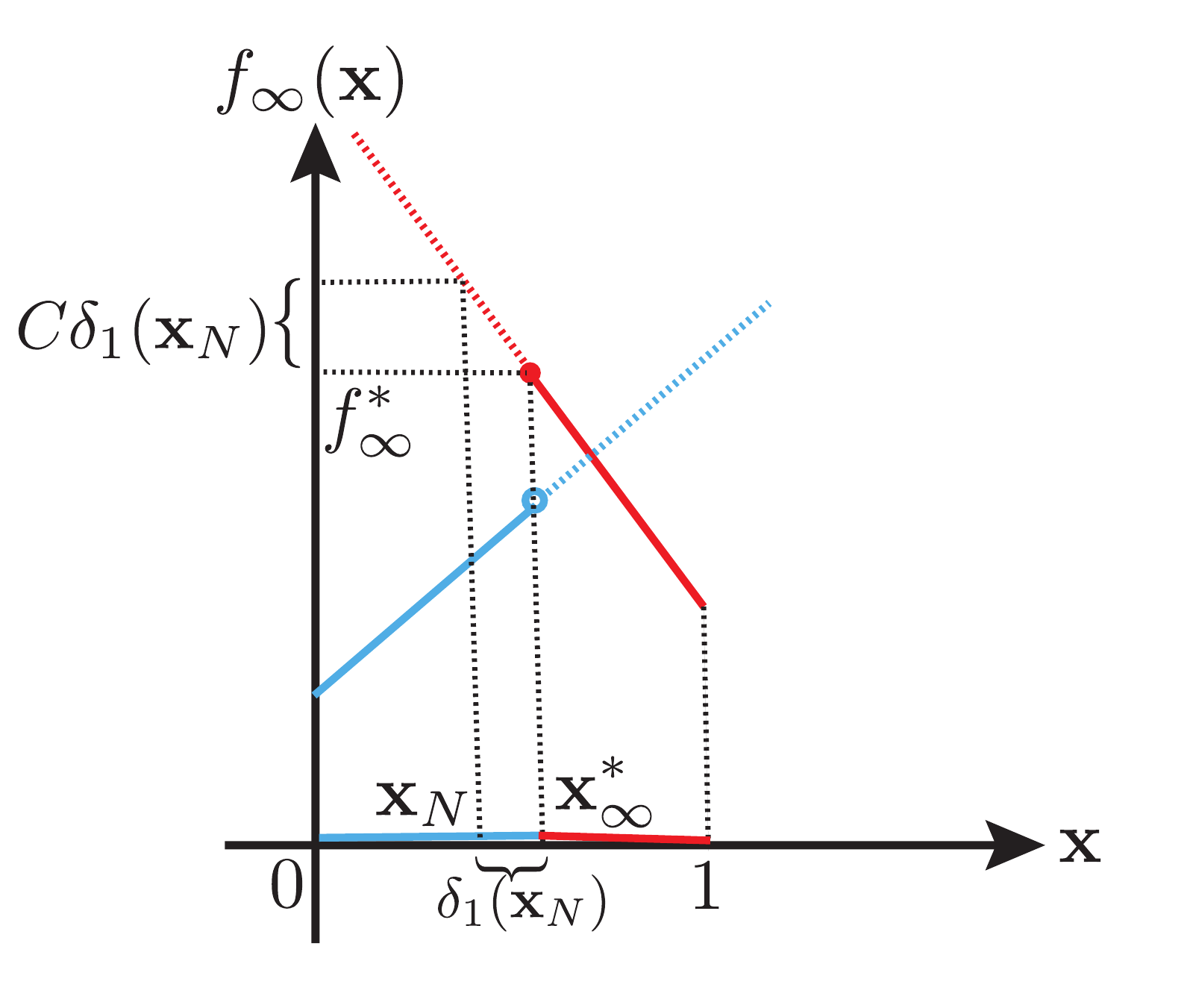}{0.5\textwidth}

Lemma~\ref{lem:boundinggain} is important because it limits the potential of leader gain from eliciting an adjacent follower response, even if she is able to do this with high probability, i.e. by committing very close to a boundary.
Figure~\ref{fig:thm3illustration2} clearly illustrates this for a $2 \times 2$ game: here, the leader might wish to elicit different-than-expected response $2$ with high probability. 
However, to do this she would have to commit close to the boundary between regions expecting responses $1$ and $2$, resulting in her payoff being close to an objective function value of $f_{\infty}(.)$ (in the figure, depicted as the optimum payoff $\fstar_{\infty}$).
For a general $m \times n$ game, the picture stays the same.

Since the quantity $\delta_1(\xvec)$ can take values anywhere in the interval $[0,2]$ (by the triangle inequality), we will still want to control the quantity $p_j(\xvec)$ for large enough values of $\delta$.
We will again use Devroye's lemma (Lemma~\ref{lem:devroye}) for tail bounding the total variation between the empirical estimate of a distribution and a true distribution.
Recall that the condition required for it to be applied was $\delta \geq \sqrt{\frac{20m}{N}}$.

It is natural to further divide the set $\Kspace^*_{\mathsf{aug}}$ into two subsets, defined by the commitment $\xvec$.
\begin{align*}
\Kspace^*_{\mathsf{aug},1}(\xvec) &:= \{j \in \Kspace^*_{\mathsf{aug}}: \delta_1(\xvec) \leq \sqrt{\frac{20m}{N}} \} \\
\Kspace^*_{\mathsf{aug},2}(\xvec) &:= \{j \in \Kspace^*_{\mathsf{aug}}: \delta_1(\xvec) > \sqrt{\frac{20m}{N}} \} .
\end{align*}

Let's consider these subsets one-by-one.
First, we use Lemma~\ref{lem:boundinggain} and the definition of the subset $\Kspace^*_{\mathsf{aug},1}(\xvec)$ to get
\begin{align}\label{eq:Tjadj1}
\sum_{j \in \Kspace^*_{\mathsf{aug},1}(\xvec)} T_j(\xvec) &= \sum_{j \in \Kspace^*_{\mathsf{aug},1}(\xvec)} p_j(x) \inprod{\bolda_j}{\xvec} \notag \\
&\leq \sum_{j \in \Kspace^*_{\mathsf{aug},1}(\xvec)} p_j(x)\left[\fstar_{\infty} + f_{max} \delta_1(\xvec;j)\right] \notag \\
&\leq \sum_{j \in \Kspace^*_{\mathsf{aug},1}(\xvec)} p_j(x)\left[\fstar_{\infty} + f_{max} \sqrt{\frac{20m}{N}}\right] \notag \\
&\leq \left[\sum_{j \in \Kspace^*_{\mathsf{aug},1}(\xvec)} p_j(x) \right] \fstar_{\infty} + f_{max} \sqrt{\frac{20m}{N}} .
\end{align}

Next, we consider the term $\sum_{j \in \Kspace^*_{\mathsf{aug},2}(\xvec)} T_j(\xvec)$.
We state and prove the following lemma.
\begin{lemma}\label{lem:Tjadj2}
For any commitment $\xvec \in \Delta_m$, we have
\begin{align}\label{eq:Tjadj2}
\sum_{j \in \Kspace^*_{\mathsf{aug},2}(\xvec)} T_j(\xvec) \leq \left[\sum_{j \in \Kspace^*_{\mathsf{aug},2}(\xvec)}p_j(\xvec)\right] \fstar_{\infty} + 3|\Kspace^*_{\mathsf{aug},2}(\xvec)| f_{max} \sqrt{\frac{20m}{N}} .
\end{align}
\end{lemma}


\begin{proof}
Consider any $j \in \Kspace^*_{\mathsf{aug},2}(\xvec)$.
Now note that by the definition of $\delta_1(\xvec;j)$, we can denote the open $\ell_1$ ball with center $\xvec$ and radius $\delta_1(\xvec)$ by $B_1(\xvec;\delta_1(\xvec))$.
By the definition of $\delta_1(\xvec;j)$, it follows that $B_1(\xvec;\delta_1(\xvec;j)) \cap \Rspace_j = \emptyset$.
Therefore, we have
\begin{align*}
p_j(\xvec) &= \Pr\left[\Xhat_N \in \Rspace_j \right] \\
&\leq \Pr\left[\Xhat_N \notin B_1(\xvec;\delta_1(\xvec;j)) \right] \\
&= \Pr\left[\vecnorm{\Xhat_N - \xvec}{1} \geq \delta_1(\xvec;j) \right] \\
&\leq 3\exp\{-\frac{N\delta_1(\xvec)^2}{25}\}
\end{align*}

where we used Lemma~\ref{lem:devroye} in the last inequality since we have $\Kspace^*_{\mathsf{aug},2}(\xvec)$, we have $\delta_1(\xvec) \geq \sqrt{\frac{20m}{N}}$.

Combining this with Lemma~\ref{lem:boundinggain}, we then have
\begin{align*}
T_j(\xvec) \leq p_j(\xvec)\fstar_{\infty} + f_{max} 3 \delta_1(\xvec;j) \exp\{-\frac{N\delta_1(\xvec;j)^2}{25}\} .
\end{align*}

Next, it is easy to verify that the function $g_2(\delta) = \delta \exp\{-\frac{N\delta^2}{25}\}$ is decreasing in $\delta$ over the domain $\delta \geq \sqrt{\frac{20m}{N}}$ for all $m \geq 1$.
This tells us that
\begin{align*}
3 \delta_1(\xvec;j) \exp\{-\frac{N\delta_1(\xvec;j)^2}{25}\} &\leq 3 \sqrt{\frac{20m}{N}} \exp\{-\frac{4m}{5}\} \\
&\leq 3 \sqrt{\frac{20m}{N}} 
\end{align*}

and so we have
\begin{align}\label{eq:Tjadj2ind}
T_j(\xvec) \leq p_j(\xvec) \fstar_{\infty} + 3 f_{max} \sqrt{\frac{20m}{N}} .
\end{align}

Summing over all $j \in \Kspace^*_{\mathsf{aug},2}(\xvec)$ and substituting Equation~\eqref{eq:Tjadj2ind} then proves the lemma.


\end{proof}



\subsubsection{Putting it all together}

Combining Equations~\eqref{eq:Tkstar},~\eqref{eq:Tjfar},~\eqref{eq:Tjadj1} and~\eqref{eq:Tjadj2} into Equation (\ref{eq:fNxvec}), we have
\begin{align*}
f_N(\xvec) &= \sum_{j=1}^n T_j(\xvec) \\
&\leq p_{\kstar}(\xvec) f^*_{\infty} + C(N+1)^m \exp\{-NC\} f_{max} + \left[\sum_{j \in \Kspace^*_{aug}} p_j(\xvec)\right] \fstar_{\infty} + (3|\Kspace^*_{\mathsf{aug},2}(\xvec)| + 1) f_{max} \sqrt{\frac{20m}{N}} \\ 
&\leq f^*_{\infty} + C(N+1)^m \exp\{-NC\} f_{max} + 4n f_{max} \sqrt{\frac{20m}{N}} \\
&\leq \fstar_{\infty} + Cn f_{max} \sqrt{\frac{20m}{N}} 
\end{align*}

for some constant $C > 0$.
This inequality holds for any $\xvec \in \Delta_m$.
This implies that $\fstar_N \leq \fstar_{\infty} + Cn\sqrt{\frac{m}{N}}$, thus completing the proof of Theorem~\ref{thm:approximation}.
\qed

\end{document}